\documentclass[a4paper, reqno]{amsart}
\usepackage[english]{babel}
\usepackage[latin1]{inputenc}
\usepackage[T1]{fontenc} 
\usepackage{graphicx}
\usepackage{amsfonts}
\usepackage{amsmath, amsthm, amssymb, yhmath, mathabx, cancel}
\usepackage{stmaryrd} 
	\SetSymbolFont{stmry}{bold}{U}{stmry}{m}{n} 
\usepackage{dsfont}
\usepackage{mathrsfs}
\usepackage{varioref}
\usepackage{color}
\usepackage{url}
\usepackage{lscape}
\usepackage{arydshln}
\usepackage{float}
\usepackage{enumitem}
\usepackage{centernot}
\usepackage[normalem]{ulem}
\usepackage{nicefrac}

\usepackage{chngcntr}

\usepackage{etoolbox}
\apptocmd{\sloppy}{\hbadness 10000\relax}{}{}

\usepackage{hyperref} 

\DeclareRobustCommand{\SkipTocEntry}[5]{}

\DeclareMathOperator{\vect}{span}

\begin{document}
\newcommand{\dd}{\,\textrm{d}}
\newcommand{\st}{\textrm{St}}

\newcommand{\R}{\mathbb{R}}
\newcommand{\Z}{\mathbb{Z}}
\newcommand{\N}{\mathbb{N}}
\renewcommand{\i}{\mathrm{i}}

\newcommand{\wto}{\rightharpoonup}

\newcommand{\pscal}[1]{{\ensuremath{\left\langle #1 \right\rangle}}}
\newcommand{\norm}[1]{{\left\vert\kern-0.25ex\left\vert #1 \right\vert\kern-0.25ex\right\vert}}
\newcommand{\normSM}[1]{{\vert\kern-0.25ex\vert #1 \vert\kern-0.25ex\vert}}
\newcommand{\normt}[1]{{\left\vert\kern-0.25ex\left\vert\kern-0.25ex\left\vert #1 \right\vert\kern-0.25ex\right\vert\kern-0.25ex\right\vert}}

\renewcommand{\phi}{\varphi}
\renewcommand{\epsilon}{\varepsilon}

\numberwithin{equation}{section}

\theoremstyle{plain}
\newtheorem{thm}{Theorem}
\newtheorem{prop}[thm]{Proposition}
\newtheorem{lemme}[thm]{Lemma}
\newtheorem{rmq}[thm]{Remark}
\theoremstyle{remark}
\newtheorem*{rmq*}{Remark}

\newcommand{\clm}[1]{\hyperref[#1]{Lemma~\ref*{#1}}}
\newcommand{\cth}[1]{\hyperref[#1]{Theorem~\ref*{#1}}}
\newcommand{\cpr}[1]{\hyperref[#1]{Proposition~\ref*{#1}}}

\title{On uniqueness and non-degeneracy of anisotropic polarons}

\author[J. Ricaud]{Julien RICAUD}
\address{CNRS, Université de Cergy-Pontoise, Mathematics Department (UMR 8088), F-95000 Cergy-Pontoise, France}
\address{CNRS, Université Paris-Dauphine, CEREMADE (UMR 7534), F-75016 Paris, France}
\email{julien.ricaud@u-cergy.fr}

\thanks{The author is grateful to M. Lewin for helpful discussions and advices. The author acknowledges financial support from the European Research Council under the European Community's Seventh Framework Programme (FP7/2007-2013 Grant Agreement MNIQS 258023), and from the French ministry of research (ANR-10-BLAN-0101).}

\date{\today}

\begin{abstract}
We study the anisotropic Choquard--Pekar equation which describes a polaron in an anisotropic medium. We prove the uniqueness and non-degeneracy of minimizers in a weakly anisotropic medium. In addition, for a wide range of anisotropic media, we derive the symmetry properties of minimizers and prove that the kernel of the associated linearized operator is reduced, apart from three functions coming from the translation invariance, to the kernel on the subspace of functions that are even in each of the three principal directions of the medium.
\end{abstract}

\maketitle

\tableofcontents

\section{Introduction}
A \emph{polaron} describes a quantum electron in a polar crystal. The atoms of the crystal are displaced due to the electrostatic force induced by the charge of the electron, and the resulting deformation is then felt by the electron itself. This coupled system (the electron and its polarization cloud) is a quasi-particle, called a polaron.

When the polaron extends over a domain much larger than the characteristic length of the underlying lattice, the crystal can be approximated by a continuous polarizable medium, leading to the so-called Pekar nonlinear model~\cite{Pekar-54,Pekar-63}. In this theory, the energy functional is
\begin{equation}\label{Pekar_energy_functional}
{\mathscr E}^V(\psi) =\frac{1}{2}\displaystyle{\int_{\R^3}}{|\nabla\psi(x)|^2\dd x} - \frac{1}{2}\displaystyle{\int_{\R^3}}{\displaystyle{\int_{\R^3}}{|\psi(y)|^2|\psi(x)|^2 V(x-y)  \dd y} \dd x},
\end{equation} where $\psi$ is the wave function of the electron, in units such that the vacuum permittivity, the mass, and the charge of the electron are all normalized to one: $4\pi\epsilon_0=m_{e^-}=e=1$. While, on the other hand, $$-V\star|\psi|^2(x)=-\displaystyle{\int_{\R^3}}{|\psi(y)|^2 V(x-y)  \dd y}$$ is the mean-field self-trapping potential felt by the electron.

For an isotropic and homogeneous medium, characterized by its \emph{relative permittivity} (or \emph{relative dielectric constant}) $\epsilon_M\geq1$,  the effective interaction potential is 
\begin{equation}\label{isotropic_effective_interaction_potential}
V(x)=\frac{1-\epsilon_M^{-1}}{|x|}.
\end{equation} For $\epsilon_M>1$, the so-called \emph{Choquard--Pekar} or \emph{Schrödinger--Newton} equation
\begin{equation}\label{Pekar_euler_lagrange_notnormalize}
\Big( -\frac{\Delta}{2} -V\star |\psi|^2 \Big)\psi=-\mu \psi
\end{equation} is obtained by minimizing the energy ${\mathscr E}^V$ in~\eqref{Pekar_energy_functional} under the constraint $\int_{\R^3}{|\psi|^2}=1$, with associated Lagrange multiplier $\mu>0$. Lieb proved in~\cite{Lieb-77} the uniqueness of minimizers, up to space translations and multiplication by a phase factor. This ground state $Q$ is positive, smooth, radial decreasing, and has an exponential decay at infinity. That $Q$ is also the unique positive solution to~\eqref{Pekar_euler_lagrange_notnormalize} was proved in~\cite{MaZha-10}.

In 2009, Lenzmann proved in~\cite{Lenzmann-09} that $Q$ is nondegenerate. Namely, the linearization of~\eqref{Pekar_euler_lagrange_notnormalize},
\begin{equation}\label{op_linearise_radial}
{\mathfrak L}_{Q}\xi=-\frac12\Delta\xi +\mu\xi -\left(V\star|Q|^2\right)\xi -2Q\left(V\star(Q\xi)\right),
\end{equation} has the trivial kernel
\begin{equation}\label{Ker_of_L_radial}
\ker_{|L^2(\R^3)} {\mathfrak L}_{Q}=\vect\left\{\partial_{x_1} Q, \partial_{x_2} Q, \partial_{x_3} Q\right\}
\end{equation} which stems from the translation invariance. This nondegeneracy result is an important property which is useful in implicit function type arguments. Uniqueness and nondegeneracy were originally used in~\cite{Lenzmann-09} to study a pseudo-relativistic model, and then in \cite{KriRapMar-09, Liu-09, RotaNodari-10,Stuart-10,FraLieSei-13, Sok-14} for other models.

The purpose of this paper is to study the case of \emph{anisotropic} media, for which the corresponding potential is
\begin{equation}\label{anisotropic_potential}
V(x)=\frac1{|x|}-\frac1{\sqrt{\det(M^{-1})}|M^{1/2} x|}, \qquad 0<M\leq1,
\end{equation} where $M^{-1}\geq1$ is the (real and symmetric) static dielectric matrix of the medium. The mathematical expression is simpler in the Fourier domain: $$\widehat{V}(k)=4\pi\left(\frac1{|k|^2}-\frac1{k^T M^{-1} k}\right).$$

The form of the potential $V$ in the anisotropic case is well-known in the physics literature and it has recently been derived by Lewin and Rougerie from a microscopic model of quantum crystals in~\cite{LewRou-11}.

From a technical point of view, the fact that $V$ in~\eqref{anisotropic_potential} is a difference of two Coulomb type potentials complicates the analysis. For this reason, we will also consider a simplified anisotropic model where $V$ is replaced by
\begin{equation}\label{anisotropic_potential_simpl}
V(x)=\frac1{|(1-S)^{-1} x|},\qquad 0\leq S<1,
\end{equation} and $S$ is also a real and symmetric  matrix. This simplified potential can be seen as an approximation of the potential~\eqref{anisotropic_potential} in the weakly anisotropic regime, that is, when $M$ is close to an homothecy.

In this paper, we derive several properties of minimizers of ${\mathscr E}^V$ and of positive solutions to the nonlinear equation~\eqref{Pekar_euler_lagrange_notnormalize}, when $V$ is given by formulas~\eqref{anisotropic_potential} and~\eqref{anisotropic_potential_simpl}. After some preparations in Section~\ref{section_settings_first_prop}, we quickly discuss the existence of minimizers and the compactness of minimizing sequences in Section~\ref{section_exist}. Then, based on the fundamental non degeneracy result by Lenzmann~\cite{Lenzmann-09}, we prove in Section~\ref{uniqueness_weakly_anisotropic} the uniqueness and non-degeneracy of minimizers in a weakly anisotropic material. In Section~\ref{section_rearr_prop}, considering back general anisotropic materials, we investigate the symmetry properties of minimizers using rearrangement inequalities. Finally we discuss the linearized operator in Section~\ref{chapter_linearized_oper}. By using Perron--Frobenius type arguments, we are able to prove that for $\psi$ a positive solution of the so-called \emph{Choquard--Pekar} equation~\eqref{Pekar_euler_lagrange_notnormalize} sharing the symmetry properties of $V$, we have
\begin{equation}\label{kernel_intro}
\ker {\mathfrak L}_\psi=\vect\left\{\partial_x \psi, \partial_y \psi, \partial_z \psi\right\} \bigoplus \ker \left({{\mathfrak L}_\psi}\right)_{|L^2_{\textrm{sym}}(\R^3)}.
\end{equation} Where $L^2_{\textrm{sym}}(\R^3)$ is the subspace of function in $L^2(\R^3)$ sharing the symmetry properties of $V$. For instance, in the general case where the three eigenvalues of $M$ (or $S$) are distinct from each other and $V$ is decreasing in the corresponding directions, $L^2_{\textrm{sym}}(\R^3)$ is the subspace of functions that are even in these directions. On the other hand, if exactly two eigenvalues are equal, it is the subspace of cylindrical functions that are also even in the directions of the principal axis.

The main difficulty in proving~\eqref{kernel_intro} is that the operator ${\mathfrak L}_\psi$ is non-local and therefore the ordering of its eigenvalues is not obvious. The next step would be to prove that $\ker {{\mathfrak L}_\psi}_{|L^2_{\textrm{sym}}(\R^3)}=\{0\}$ which we only know for now in the weakly anisotropic regime (\cth{Weakly_anisotr_uniqueness_minim} below) and in the radial case (see~\cite{Lenzmann-09}). We hope to come back to this problem in the future.

\section{Elementary properties}\label{section_settings_first_prop}
We define the energy ${\mathscr E}^V$ as in~\eqref{Pekar_energy_functional} and consider, for all $\lambda>0$, the minimization problem 
\begin{equation}\label{minimization_pb}
I^V(\lambda):=\min\limits_{\substack{\psi \in H^1(\R^3) \\\norm{\psi}^2_2=\lambda}}{\mathscr E}^V(\psi).
\end{equation}

Let $(e_1,e_2,e_3)$ be the principal axis of the medium, that is, such that each $e_i\in\R^3$ is a normalized eigenvector associated with the eigenvalue $m_i$ of the real symmetric matrix $M$, where $0< m_1\leq m_2\leq m_3\leq 1$ (we demand, in addition, that at least $m_2<1$), or associated with the eigenvalue $s_i$ of the real symmetric matrix $S$ where $0\leq s_3\leq s_2\leq s_1< 1$ in the simplified model.

We define the map $M\mapsto V$ as
\begin{equation}\label{def_M_to_V}
\begin{aligned}
\left\{0< M\leq1\; |\; M\textrm{ symmetric real}\right\} &\to L^2(\R^3) + L^4(\R^3)\\
M&\mapsto V(x)=\frac1{|x|}-\frac1{\sqrt{\det(M^{-1})}|M^{1/2} x|}
\end{aligned}
\end{equation} with, in particular, $M\equiv\textrm{Id}\mapsto V\equiv\bar{0}$ which corresponds to the vacuum. And, in the simplified model, $S\mapsto V$ is defined as
\begin{equation}\label{def_S_to_V}
\begin{aligned}
\left\{0\leq M<1\; |\; M\textrm{ symmetric real}\right\} &\to L^2(\R^3) + L^4(\R^3)\\
S&\mapsto V(x)=|(1-S)^{-1}x|^{-1}
\end{aligned}
\end{equation} with, in particular, $S\equiv0\mapsto V\equiv V_0$. We denote the isotropic potentials by $V_c(x)=(1-c)|x|^{-1}$, for $0\leq c\leq1$, and $I^V_c$ the associated minimization problem.

Both maps are well-defined. Indeed, let $V$ as in~\eqref{def_M_to_V} or~\eqref{def_S_to_V} then one can easily show that there exist $a>b\geq0$ such that
\begin{equation}\label{aniso_propri_V_large}
\forall x \in \R^3_*,\; 0\leq b|x|^{-1}\leq V(x) \leq a|x|^{-1} \leq |x|^{-1}.
\end{equation} Consequently, $V \in L^2(\R^3) + L^4(\R^3)$. Moreover, if we restrict ourselves to $0< M<1$ then there exist $a>b>0$ such that
\begin{equation}\label{aniso_propri_V}
\forall x \in \R^3_*,\; 0< b|x|^{-1}\leq V(x) \leq a|x|^{-1} \leq |x|^{-1}.
\end{equation}

\begin{lemme}\label{majoration_(fg)star-(V-W)}Let $M\mapsto V$ be defined as in \eqref{def_M_to_V}, $S\mapsto V$ as in~\eqref{def_S_to_V} and let $f,g$ be two functions in $H^1(\R^3)$. Then $V\star(fg)\in W^{1,\infty}$ and, for any $0<\alpha<1$, we have
\begin{enumerate}
	\item locally Lipschitzity of
	\begin{align*}
\left\{\alpha< M\leq1\; |\; M\textrm{ symmetric real}\right\}\times H^1\times H^1 &\to W^{1,\infty}\\
(M,f,g)&\mapsto V\star(fg),
\end{align*}
	\item uniform Lipschitzity of
	\begin{align*}
\left\{0\leq S<\alpha\; |\; S\textrm{ symmetric real}\right\}\times H^1\times H^1 &\to W^{1,\infty}\\
(S,f,g)&\mapsto V\star(fg).
\end{align*}
\end{enumerate}
\end{lemme}
\begin{proof}[Proof of~\clm{majoration_(fg)star-(V-W)}]First, for any $f\in L^2(\R^3)$ and $g\in H^1(\R^3)$, by \eqref{aniso_propri_V_large} together with Hardy inequality, $\left|V\star(fg)(x)\right|\leq (|\cdot|^{-1}\star |fg|)(x)\leq 2\norm{f}_2\norm{\nabla g}_2$ holds. Consequently, for any $f,g\in H^1(\R^3)$, we have
\begin{align*}
\norm{V\star(fg)}_{W^{1,\infty}}&\leq \norm{V\star(fg)}_\infty + \norm{V\star(g \nabla f)}_\infty + \norm{V\star(f \nabla g)}_\infty \\
	&\leq 2\norm{f}_2\norm{\nabla g}_2 + 4\norm{\nabla f}_2\norm{\nabla g}_2 \\
	&\leq 6\norm{f}_{H^1}\norm{g}_{H^1}.
\end{align*} Thus $V\star(fg)$ is in $W^{1,\infty}$. For the rest of the proof, we denote by $\norm{M}$ the spectral norm of $M$ and fix an $\alpha$ such that $0<\alpha<1$.

For $(S,T) \in \left\{0\leq M<\alpha\; |\; M\textrm{ symmetric real}\right\}^2$, $f\in L^2(\R^3)$, $g\in H^1(\R^3)$ and $x\in\R^3$, we have
\begin{align*}
\left|(V_S-V_T)\star(fg)(x)\right|&\leq \left|\frac{|(1-T)^{-1}\cdot|-|(1-S)^{-1}\cdot|}{|(1-S)^{-1}\cdot||(1-T)^{-1}\cdot|}\right| \star |fg| (x)\\
	&\leq \frac{\left|\left[(1-T)^{-1}-(1-S)^{-1}\right]\cdot\right|}{|\cdot|^2} \star |fg| (x)\\
	&\leq \frac{\left|(1-S)^{-1}(T-S)(1-T)^{-1}\cdot\right|}{|\cdot|^2} \star |fg| (x)\\
	&\leq \norm{(1-S)^{-1}}\norm{T-S}\norm{(1-T)^{-1}} \frac1{|\cdot|} \star |fg| (x)\\
	&\leq 2(1-\alpha)^{-2}\norm{f}_2\norm{\nabla g}_2\norm{S-T}.
\end{align*} Thus, for any $f,g\in H^1(\R^3)$, we have $$\norm{(V_S-V_T)\star(fg)}_{W^{1,\infty}}\leq 6(1-\alpha)^{-2}\norm{f}_{H^1}\norm{g}_{H^1}\norm{S-T},$$ which concludes the proof of \emph{(2)}.

For $(M,N) \in \left\{\alpha< M\leq1\; |\; M\textrm{ symmetric real}\right\}^2$, we have
\begin{align*}
M^{1/2}-N^{1/2}&=\pi^{-1}\displaystyle{\int_o^\infty}{\left(\frac{M}{s+M}-\frac{N}{s+N}\right)\frac{\dd s}{\sqrt{s}} }\\
	&=\pi^{-1}\displaystyle{\int_o^\infty}{\frac1{s+M}(M-N)\frac1{s+N} \sqrt{s}\dd s },
\end{align*} which leads to $$\norm{M^{\frac12}-N^{\frac12}}\leq\frac{\norm{M-N}}{\pi}\displaystyle{\int_o^\infty}{\frac{\sqrt{s}}{(s+\alpha)^2} \dd s }=\frac{\norm{M-N}}{\pi\sqrt{\alpha}}\displaystyle{\int_o^\infty}{\frac{\sqrt{s}}{(s+1)^2} \dd s }=\frac{\norm{M-N}}{2\sqrt{\alpha}}.$$ Moreover, with a similar computation and since $\det M, \det N>\alpha^3$, we obtain $$\left|\sqrt{\det M}-\sqrt{\det N}\right|\leq\frac{\left|\det M-\det N\right|}{2\alpha^{3/2}}.$$ Thus, for $f\in L^2(\R^3)$, $g\in H^1(\R^3)$ and $x\in\R^3$, we have
\begin{align*}
\left|(V_{M}-V_{N})\star(fg)(x)\right|&\leq \frac1{\sqrt{\det N^{-1}}} |fg|\star\frac{\left|\left|M^{1/2}\cdot\right|-\left|N^{1/2}\cdot\right|\right|}{\left|M^{1/2}\cdot\right|\left|N^{1/2}\cdot\right|}(x) \\
	&\phantom{\leq} +\left|\frac1{\sqrt{\det N^{-1}}}-\frac1{\sqrt{\det M^{-1}}}\right| |fg|\star\left|M^{1/2}\cdot\right|^{-1}(x)\\
	&\leq 2\sqrt{\det N}\norm{M^{-1}}^{1/2}\norm{N^{-1}}^{1/2} \norm{f}_2\norm{\nabla g}_2 \norm{M^{1/2}-N^{1/2}} \\
	&\phantom{\leq} +2\norm{M^{-1}}^{1/2}\norm{f}_2\norm{\nabla g}_2 \left|\sqrt{\det N}-\sqrt{\det M}\right|\\
	&\leq \left(\norm{M-N}+\alpha^{-1/2}\left|\det N-\det M\right|\right)\alpha^{-3/2} \norm{f}_2\norm{\nabla g}_2.
\end{align*} Finally, the determinant being locally Lipschitz, we obtain that $M \mapsto V\star(fg)$ is locally Lipschitz.
\end{proof}

Since $M^{-1}$ is real and symmetric, there exists $R \in O(3)$ such that $$R^TMR=\textrm{diag}(m_3,m_2,m_1)$$ and so, for any $x\in\R^3$, after a simple computation, we have
$$V(Rx)=|x|^{-1}-\left|\textrm{diag}\left((m_1m_2)^{-1/2},(m_1m_3)^{-1/2},(m_2m_3)^{-1/2}\right)x\right|^{-1},$$ where $0<\sqrt{m_1m_2}\leq\sqrt{m_1m_3}\leq\sqrt{m_2m_3}<1$ since $m_2<1$. Thus,
we can consider, without any loss of generality, that
\begin{equation}\label{V_final_definition}
	\left\{
		\begin{aligned}
			M&=\textrm{diag}(m_1,m_2,m_3) \quad\textrm{ with } 0< m_3\leq m_2\leq m_1<1,\\
			M&\mapsto V(x)=\frac1{|x|}-\frac1{|M^{-1} x|}
		\end{aligned}
	\right.
\end{equation}

Similarly, for the simplified model, we can also assume that
\begin{equation}\label{V_simplified_final_definition}
V(x)=|\textrm{diag}(1-s_1,1-s_2,1-s_3)^{-1}x|^{-1}, \qquad 0\leq s_3\leq s_2\leq s_1<1.
\end{equation}

For clarity, from now on we denote by ${\mathscr E}_M$ (resp. ${\mathscr E}_S$) the energy and by $I_M(\lambda)$ (resp. $I_S(\lambda)$) the minimization problem since $V$ depends only on the matrix $M$ (resp. on the matrix $S$). However, for shortness, we will omit the subscripts when no confusion is possible.

\begin{lemme}\label{symmetric_also_minimizer}
Let $\psi\in H^1(\R^3)$ be a solution of the equation~\eqref{Pekar_euler_lagrange_notnormalize}, for $V$ defined as in~\eqref{V_final_definition} or in~\eqref{V_simplified_final_definition}, then $(x,y,z)\mapsto\psi(\pm x,\pm y,\pm z)$ are $H^1(\R^3)$-solutions to~\eqref{Pekar_euler_lagrange_notnormalize}.
\end{lemme}
\begin{proof}[Proof of~\clm{symmetric_also_minimizer}]
This follows from the symmetry properties of $V$.
\end{proof}

\section{Existence of minimizers}\label{section_exist}
We state in this section the existence of minimizers for the minimization problems (see Appendix in~\cite{Ricaud-PhD} for the proof). We first give some properties of these variational problems.
\begin{lemme}\label{lemma_ppt_I_anisotro} Let $V$ be defined as in~\eqref{V_final_definition} or~\eqref{V_simplified_final_definition} and $I$ be defined as in~\eqref{minimization_pb}. Then
\begin{equation}\label{I_neg_anisotro}
I(\lambda)=\lambda^3I(1) < 0, \textrm{ if } \lambda>0.
\end{equation}
Consequently,
\begin{enumerate}
	\item $\lambda \mapsto I(\lambda)$ is $C^\infty$ on $\R^+$,
	\item $I(\lambda) < I(\lambda-\lambda')+I(\lambda')$, for any $\lambda$ et $\lambda'$ such that $0 < \lambda' < \lambda$,
\end{enumerate}
and, in particular:
\begin{enumerate}
	\setcounter{enumi}{2}
	\item $I(\lambda) < I(\lambda')$, for any $0 \leq \lambda' < \lambda$.
\end{enumerate}
\end{lemme} 
\begin{proof}[Proof of \clm{lemma_ppt_I_anisotro}] Let $\psi \in H^1(\R^3)$ with $\norm{\psi}_{L^2}^2=1$, then $\psi_\lambda:=\lambda^2\psi(\lambda \cdot) \in H^1(\R^3)$ and $\norm{\psi_\lambda}_{L^2}^2=\lambda$ and, by a direct computation, ${\mathscr E}(\psi_\lambda)=\lambda^3{\mathscr E}(\psi)$ which leads to $I(\lambda)=\lambda^3I(1)$. If we now define $\psi_t = t^{3/2}\psi(t\cdot)$ and use \eqref{aniso_propri_V}, we find that 
 \begin{align*}
 {\mathscr E}(\psi_t) &\leq \frac{1}{2} \norm{\nabla\psi_t}_{L^2}^2 - \frac{b}{2} \norm{|\psi_t|^2 \left(|\psi_t|^2\star|\cdot|^{-1}\right)}_{L^2}^2 \\
	&\leq\frac{t^2}{2}\norm{\nabla\psi}_{L^2}^2 - \frac{b t}{2}\norm{|\psi|^2\left(|\psi|^2\star|\cdot|^{-1}\right)}_{L^2}^2,
 \end{align*} and taking $t$ small enough leads to the claimed strict negativity. The rest follows immediately.
\end{proof}

\begin{lemme}\label{anisot_convex_I}
Let  $V$ be defined as in~\eqref{V_final_definition} or~\eqref{V_simplified_final_definition}. Let $I$ be as in~\eqref{minimization_pb} and let $\lambda > 0$. Then $I(t\lambda) > t I(\lambda), \; \textrm{for all } t \in (0,1)$.
\end{lemme}
\begin{proof}[Proof of~\clm{anisot_convex_I}]
Let $t \in (0,1)$. By  \clm{lemma_ppt_I_anisotro}, $0>I(t\lambda)=t^3I(\lambda)>tI(\lambda)$.
\end{proof}

These two lemmas imply the existence of minimizers and the compactness of minimizing sequences, as stated in the following theorem which gives also some properties of these minimizers.
\begin{thm}[Existence of a minimizer]\label{anisotrop_existence_minimiseur} Let $V$ be as in~\eqref{V_final_definition} or~\eqref{V_simplified_final_definition} and $\lambda>0$. Then $I(\lambda)$ has a minimizer and any minimizing sequence strongly converges in $H^1(\R^3)$ to a minimizer, up to extraction of a subsequence and after an appropriate space translation.

Moreover for any minimizer $\psi$, we have
\begin{enumerate}
	\item $\psi$ is a $H^2(\R^3)$-solution of the Choquard--Pekar equation~\eqref{Pekar_euler_lagrange_notnormalize}\\with $-\mu=\frac{d}{d\lambda}I(\lambda)<0$ being the smallest eigenvalue of the self-adjoint operator $H_{\psi}:=-\Delta/2 - |\psi|^2\star V$ and being simple;
	\item \hfill	
	\makebox[0pt][r]{%
		\begin{minipage}[b]{\textwidth}
			\begin{equation}\label{anisotrop_existence_minimiseur_norm_equal}
				\mu\lambda=-\lambda \frac{d}{d\lambda}I(\lambda)=-3\lambda^3I(1)=\frac32\norm{\nabla\psi}_2^2=\frac34\pscal{V\star|\psi|^2, |\psi|^2};
			\end{equation}
		\end{minipage}
	}		
	\item $|\psi|$ is a minimizer and $|\psi|>0$;
	\item $\psi = z|\psi|$ for a given $|z|=1$.
\end{enumerate}
\end{thm}

See Appendix in~\cite{Ricaud-PhD} for the proof. Note that for the isotropic potentials $V_c$, Lieb proved several of these results in 1977~\cite{Lieb-77} using only the fact that $|x|^{-1}$ is radially decreasing.

\section{Uniqueness in a weakly anisotropic material}\label{uniqueness_weakly_anisotropic}
We recall that the uniqueness of the minimizer, up to phases and space translations, in the isotropic case, was proven by Lieb in~\cite{Lieb-77}. In this section, we extend this result to the case of \emph{weakly anisotropic materials}, meaning that we consider static dielectric matrices close to an homothecy.

We first prove the continuity of $I_M(\lambda)$, with respect to $(M,\lambda)$, which we will need in the proof of uniqueness.
\begin{lemme}[Minimums' convergence]\label{Conv_minimum_I_with_lambda_M}
Let $V$ be defined as in~\eqref{V_final_definition} or~\eqref{V_simplified_final_definition}, $I$ be defined as in~\eqref{minimization_pb} and $(\lambda,\lambda')\in\left(\R_+^*\right)^2$. Then $$I_{M'}(\lambda') \underset{\substack{\norm{M'-M}\rightarrow 0 \\ |\lambda'-\lambda|\rightarrow 0}}{\xrightarrow{\hspace*{1.5cm}}} I_M(\lambda).$$
Thus, the continuity of the corresponding Euler-Lagrange multiplier, $-\mu_{M',\lambda'}$, holds as well: $$\mu_{M',\lambda'} \underset{\substack{\norm{M'-M}\rightarrow 0 \\ |\lambda'-\lambda|\rightarrow 0}}{\xrightarrow{\hspace*{1.5cm}}} \mu_{M,\lambda}.$$
\end{lemme}
\begin{proof}[Proof of~\clm{Conv_minimum_I_with_lambda_M}]Let $\psi$ (resp. $\psi'$) be a minimizer of $I_M(\lambda)$ (resp. $I_{M'}(\lambda)$) for a given $\lambda>0$.

First, for any $\phi\in H^1(\R^3)$, we have $$\left|{\mathcal E}_M(\phi)-{\mathcal E}_{M'}(\phi)\right|=\frac12\left|\pscal{|\phi|^2,|\phi|^2\star(V-V')}\right|\leq\frac12\norm{|\phi|^2\star(V-V')}_\infty\norm{\phi}_2^2.$$ Thus, by~\clm{majoration_(fg)star-(V-W)}, $M\mapsto{\mathcal E}_M(\phi)$ is Lipschitz for any $\phi\in H^1(\R^3)$. Moreover $${\mathcal E}_M(\psi)-{\mathcal E}_{M'}(\psi) \leq I_M(\lambda) - I_{M'}(\lambda) \leq {\mathcal E}_M(\psi')-{\mathcal E}_{M'}(\psi'),$$ which implies that $M\mapsto I_M(\lambda)$ is Lipschitz for any $\lambda>0$.

Thanks to~\clm{lemma_ppt_I_anisotro}, we conclude the proof of the convergence of $I$ since $$\left|I_M(\lambda) - I_{M'}(\lambda')\right| \lesssim \left|I_M(1)\right| \left|\lambda^3 - (\lambda')^3\right| + \norm{M-M'}.$$ Then, the equality $-\mu_{M,\lambda}=3\lambda^2I_M(1)$ gives the convergence of the $\mu_{M',\lambda'}$'s.
\end{proof} We now give our theorem of uniqueness in the weakly anisotropic case.

\begin{thm}[Uniqueness and non-degeneracy in the weakly anisotropic case]\label{Weakly_anisotr_uniqueness_minim}\hspace{0.1cm}\\
Let $\lambda>0$.
\begin{enumerate}[leftmargin=0.5cm, label=\roman*.]
	\item Let $0 < s < 1$. There exists $\epsilon>0$ such that, for every real symmetric $3\times3$ matrix $0<M<1$ with $\norm{M-s\cdot\textrm{\emph{Id}}} <\epsilon$, the minimizer $\psi$ of the minimization problem $I_M(\lambda)$, for $V(x)=|x|^{-1}-|M^{-1} x|^{-1}$ as in~\eqref{V_final_definition}, is unique up to phase and space translations.
	\item Let $0 \leq s < 1$. There exists $\epsilon>0$ such that, for every real symmetric $3\times3$ matrix $0\leq S<1$ with $\norm{S-s\cdot\textrm{\emph{Id}}} <\epsilon$, the minimizer $\psi$ of the minimization problem $I_S(\lambda)$, for $V(x)=|(1-S)^{-1}x|^{-1}$ as in~\eqref{V_simplified_final_definition}, is unique up to phase and space translations.
\end{enumerate}
Moreover, in both cases, the minimizer is even along each eigenvectors of $M$ and $\ker {\mathfrak L}_\psi=\vect\left\{\partial_x \psi, \partial_y \psi, \partial_z \psi\right\}$, where ${\mathfrak L}_\psi$ is the linearized operator defined in~\eqref{op_linearise_radial}.
\end{thm}

The proof of this theorem is based on a perturbative argument around the isotropic case, using the implicit functions theorem. The fundamental nondegeneracy result in the isotropic case, proved by Lenzmann in~\cite{Lenzmann-09}, is a key ingredient of our proof.

\begin{proof}[Proof of \cth{Weakly_anisotr_uniqueness_minim}] The proof of $ii$ being similar to the one of $i$, we will only give the latter. Let us fix $0 < s < 1$, define ${\mathscr D}:=\left\{0< M<1\; |\; M\textrm{ symmetric real}\right\}$ and denote by $Q$ the unique positive minimizer of the isotropic minimization problem $I(\lambda):=I_{s\cdot\textrm{Id}}(\lambda)$ for $V(x)=V_{s\cdot\textrm{Id}}(x)=(1-s)|x|^{-1}$, which is radial and solves~\eqref{Pekar_euler_lagrange_notnormalize}:
\begin{equation*}
-\frac1{2}\Delta Q +\mu Q - (|Q|^2\star V)Q=0,
\end{equation*} with $\norm{Q}^2_2=\lambda$. There $\lambda$ is fixed hence is $\mu:=\mu_{s\cdot\textrm{Id},\lambda}>0$ by~\clm{lemma_ppt_I_anisotro}.

\addtocontents{toc}{\SkipTocEntry} 
\subsection*{Step 1: Implicit function theorem and local uniqueness.}
By Proposition 5 in~\cite{Lenzmann-09}, we know that the linearized operator ${\mathfrak L}_Q$ given by
\begin{equation}\label{Linearized_L_radial}
{\mathfrak L}_Q\xi = -\frac12\Delta\xi + \mu\xi - \left(V \star |Q|^2\right)\xi - 2Q\left(V \star (Q\xi)\right),
\end{equation}
acting on $L^2(\R^3)$ with domain $H^2(\R^3)$, has the kernel
\begin{equation}\label{Kernel_L_radial}
\ker {\mathfrak L}_Q=\vect\left\{\partial_{x_1}Q,\partial_{x_2}Q,\partial_{x_3}Q\right\}.
\end{equation}
Let us define $u$ as
\begin{equation*}
\begin{aligned}
H^1(\R^3,\R) \times {\mathscr D} &\overset{u}{\to} L^2(\R^3,\R) \\
(\psi,M)&\mapsto - \left(|\psi|^2\star V\right)\psi
\end{aligned}
\end{equation*}
 and $G$ as
\begin{equation*}
\begin{aligned}
(\ker {\mathfrak L}_Q)^\perp \times {\mathscr D} &\overset{G}{\to} H^1(\R^3,\R) \\
(\psi,M)&\mapsto\psi + \left(-\Delta/2 + \mu_M\right)^{-1}u(\psi,M),
\end{aligned}
\end{equation*} where $(\ker {\mathfrak L}_Q)^\perp$ is the orthogonal of $\ker {\mathfrak L}_Q$ for the scalar product of $L^2(\R^3)$, which we endow with the norm of $H^1(\R^3)$, and $\mu_M:=\mu_{M,\lambda}=3\lambda^2 I_M(1)$. We emphasize here that we consider real valued functions, meaning that we are constructing a branch of real valued solutions. Moreover, $G(\psi,M)=0$ is equivalent to $-\frac1{2}\Delta \psi +\mu \psi - (|\psi|^2\star V)\psi=0$. Differentiating with respect to $x_i$, for $i=1,2,3$, we get ${\mathfrak L}_\psi\partial_{x_i}\psi=0$, for $i=1,2,3$, and thus $\vect\left\{\partial_x \psi, \partial_y \psi, \partial_z \psi\right\}\subset\ker {\mathfrak L}_\psi$.

By the Hardy-Littlewood-Sobolev and Sobolev inequalities, $u$ is well defined. Moreover, splitting $u(\psi,M) - u(\psi',M')$ into three pieces and using~\eqref{aniso_propri_V} together with the Hardy inequality, one obtains
\begin{align*}
\norm{u(\psi,M) - u(\psi',M')}_{L^2}
&\leq \norm{V\star|\psi|^2}_{L^\infty}\norm{\psi-\psi'}_{L^2} + \norm{(V-V')\star|\psi|^2}_{L^\infty}\norm{\psi'}_{L^2} \\
&\phantom{\leq} + \norm{V'\star\left(|\psi|^2-|\psi'|^2\right)}_{L^\infty}\norm{\psi-\psi'}_{L^2} \\
&\leq \norm{V\star|\psi|^2}_{L^\infty}\norm{\psi-\psi'}_{L^2} + \norm{(V-V')\star|\psi|^2}_{L^\infty}\norm{\psi'}_{L^2} \\
&\phantom{\leq} + 2\left(\norm{\psi}_{H^1}+\norm{\psi'}_{H^1}\right)\norm{\psi-\psi'}_{L^2}^2.
\end{align*}
Therefore, using~\clm{majoration_(fg)star-(V-W)}, $u$ is locally Lipschitz on $H^1(\R^3,\R) \times {\mathscr D}$. Then, since $\left(-\Delta/2 + \mu_M\right)^{-1}$ maps $L^2(\R^3)$ onto $H^2(\R^3)\subset H^1(\R^3)$, $G$ is also well defined. Moreover, since $\normt{(-\Delta+\nu)^{-1}}_{L^2\to H^2}\leq\max\{1,\nu^{-1}\}$ (for $\nu>0$) and $$\left(-\Delta/2 + a\right)^{-1} - \left(-\Delta/2 + b\right)^{-1}=(b-a)\left(-\Delta/2 + a\right)^{-1}\left(-\Delta/2 + b\right)^{-1},$$ for all $a,b>0$, we have
\begin{multline*}
	\norm{G(\psi,M) - G(\psi',M')}_{H^1}\\
	\begin{aligned}
   		&\leq \norm{\psi-\psi'}_{H^1} + \norm{\left(-\Delta/2 + \mu_M\right)^{-1}\left(u(\psi,M)-u(\psi',M')\right)}_{H^1} \\
		&\phantom{\leq} +\left|\mu_{M'}-\mu_M\right|\norm{\left(-\Delta/2 + \mu_M\right)^{-1}\left(-\Delta/2 + \mu_{M'}\right)^{-1}u(\psi',M')}_{H^1}\\
		&\lesssim \norm{\psi-\psi'}_{H^1} + \max\{2,(\mu_M)^{-1}\}\norm{u(\psi,M)-u(\psi',M')}_{L^2} \\
		&\phantom{\leq} +\max\{2,(\mu_M)^{-1}\}\max\{2,(\mu_{M'})^{-1}\}\norm{u(\psi',M')}_{L^2}\norm{M'-M},
   	\end{aligned}
\end{multline*} which proves that $G$ is also locally Lipschitz.

A simple computation shows that
\begin{equation}\label{diff_u}
\partial_\psi u(\psi,M)\xi=- \left(|\psi|^2\star V\right)\xi - 2\psi\left((\psi\xi)\star V\right),
\end{equation} acting on $\xi\in (\ker {\mathfrak L}_Q)^\perp$,  and that
\begin{equation}\label{dif_partialf_G}
\partial_\psi G(\psi,M)=1+\left(-\Delta/2 + \mu_M\right)^{-1}\partial_\psi u(\psi,M).
\end{equation} We claim $\partial_\psi G(\phi,M)$, defined from $(\ker {\mathfrak L}_Q)^\perp \times {\mathscr D}$ into ${\mathcal L}\left((\ker {\mathfrak L}_Q)^\perp,L^2(\R^3,\R)\right)$, to be continuous. Indeed, $$\norm{\partial_\psi u(\psi,M)\xi}_{L^2}\leq \norm{\xi}_{L^2}\norm{V\star |\psi|^2}_{L_\infty}+2\norm{\psi}_{L^2}\norm{(\psi\xi)\star V}_{L^\infty}\leq3\norm{\psi}_{H^1}\norm{\psi}_{L^2} \norm{\xi}_{L^2},$$ thus $\partial_\psi u(\psi,M)\xi\in L^2(\R^3,\R)$ for any $(\phi,M,\xi)\in (\ker {\mathfrak L}_Q)^\perp \times {\mathscr D}\times (\ker {\mathfrak L}_Q)^\perp$.

Splitting again the term into pieces and using~\eqref{aniso_propri_V}, for $\xi\in L^2(\R^3,\R)$, one obtains
\begin{multline*}
	\norm{\partial_\psi u(\psi,M)\xi- \partial_1 u(\psi',M')\xi}_{L^2}\\
	\begin{aligned}
   		&\leq \norm{V \star \left(|\psi|^2 - |\psi'|^2\right)}_{L^\infty}\norm{\xi}_{L^2} + \norm{\left(V - V'\right) \star |\psi'|^2}_{L^\infty}\norm{\xi}_{L^2} \\
		&\phantom{\leq} + 2\norm{V \star (\psi\xi)}_{L^\infty}\norm{\psi - \psi'}_{L^2} + 2\norm{V \star ((\psi - \psi')\xi)}_{L^\infty}\norm{\psi'}_{L^2} \\
		&\phantom{\leq} + 2\norm{(V - V') \star (\psi'\xi)}_{L^\infty}\norm{\psi'}_{L^2} \\
		&= O\left(\norm{(\psi,M) - (\psi',M')}_{H^1 \times {\mathscr D}}\right)\norm{\xi}_{L^2}.
   	\end{aligned}
\end{multline*} Then, since
\begin{multline*}
	\norm{\partial_\psi G(\psi,M)\xi - \partial_{\psi} G(\psi',M')\xi}_{H^1}\\
	\begin{aligned}
		&\lesssim \max\{2,(\mu_M)^{-1}\}\norm{\partial_\psi u(\psi,M)-\partial_\psi u(\psi',M')}_{L^2} \\
		&\phantom{\leq} +\max\{2,(\mu_M)^{-1}\}\max\{2,(\mu_{M'})^{-1}\}\norm{\partial_\psi u(\psi',M')}_{L^2}\norm{M'-M},
   	\end{aligned}
\end{multline*} we have $$\normt{\partial_\psi G(\psi,M) - \partial_{\psi} G(\psi',M')} \rightarrow 0, \textrm{ if } \norm{(\psi,M)-(\psi',M')}_{H^1 \times {\mathscr D}} \rightarrow 0.$$ This concludes the proof of the continuity of $\partial_\psi G(\phi,M)$ from $(\ker {\mathfrak L}_Q)^\perp \times {\mathscr D}$ into ${\mathcal L}\left((\ker {\mathfrak L}_Q)^\perp,H^1(\R^3,\R)\right)$.

We now apply the implicit function theorem to $G$. Indeed, by the definition of $(\ker {\mathfrak L}_Q)^\perp$, the restriction of ${\mathfrak L}_Q$ to $(\ker {\mathfrak L}_Q)^\perp$ has a trivial kernel. On the other hand, the operator $\left(-\Delta/2 + \mu_M\right)^{-1}\partial_\psi u(Q,s\cdot\textrm{Id})$ is compact on $L^2(\R^3)$ (see Appendix in~\cite{Ricaud-PhD}), therefore $-1$ does not belong to its spectrum. We deduce from this the existence of the inverse operator
\begin{equation}\label{inverse_diff_part_psi}
\left(\partial_\psi G(Q,s\cdot\textrm{Id})\right)^{-1} : \textrm{Ran}(G)\subset H^1(\R^3,\R)\to (\ker {\mathfrak L}_Q)^\perp.
\end{equation} Then, by the continuity of $G$ and $\partial_\psi G$, the existence of $\left(\partial_\psi G(Q_s,s\cdot\textrm{Id})\right)^{-1}$ and since $G(Q,s\cdot\textrm{Id})=0$, the inverse function theorem 1.2.1 of~\cite{Cha05} implies that there exist $\delta,\epsilon>0$ such that there exists a unique $\psi(M)\in (\ker {\mathfrak L}_Q)^\perp$ satisfying:
\begin{equation}\label{local_uniqueness_radial_implicit}
G(\psi(M),M)=0 \quad \textrm{for } \norm{M-s\cdot\textrm{Id}}\leq\epsilon \textrm{ and }\norm{\psi(M)-Q}_{H^1}\leq \delta.
\end{equation} Moreover, the map $M\mapsto \psi(M)$ is continuous.

Additionally, $\ker \partial_\psi G(\psi(M),M)=\{0\}$, i.e. $\ker_{|(\ker {\mathfrak L}_Q)^\perp} {\mathfrak L}_\psi=\{0\}$ which leads to $\dim \ker\left({\mathfrak L}_\psi\right)\leq3$ since $\dim \ker\left({\mathfrak L}_Q\right)=3$ by~\eqref{Kernel_L_radial}.

We now claim that $\psi(M)$ is symmetric with respect to the three eigenvectors of $M$, $\{e_i\}_{i=1,2,3}$, and consequently that, for $i=1,2,3$, $\partial_{x_i} \psi(M)$ is odd along $e_i$ and even along $e_j$ for $j\neq i$. Indeed $V$ being symmetric, the eight functions $(x,y,z)\mapsto\psi(M)(\pm x, \pm y, \pm z)$, which are in $(\ker {\mathfrak L}_Q)^\perp$, are zeros of $G(\cdot,M)$. If $\psi(M)$ were not symmetric with respect to each $e_i$, then at least two of the functions $(x,y,z)\mapsto\psi(M)(\pm x, \pm y, \pm z)$ would be distinct functions but both verifying~\eqref{local_uniqueness_radial_implicit}, since $Q$ is symmetric with respect to each $e_i$, which is impossible by local uniqueness.

Thus the $\partial_{x_i} \psi(M)$'s are orthogonal and we have $\dim \vect\left\{\partial_x \psi, \partial_y \psi, \partial_z \psi\right\}=3$. Since $\vect\left\{\partial_x \psi, \partial_y \psi, \partial_z \psi\right\}\subset\ker {\mathfrak L}_\psi$, this leads to $\dim \ker\left({\mathfrak L}_\psi\right)\geq3$. Which proves that $\ker {\mathfrak L}_\psi=\vect\left\{\partial_x \psi, \partial_y \psi, \partial_z \psi\right\}$.

Let us emphasize that, at this point, we do not know the masses $\norm{\psi(M)}_2^2$ of those $\psi(M)$. Note also that we could prove here that $|\psi|>0$, since $-\mu_M$ stays the first eigenvalue by continuity and with a Perron--Frobenius type argument, but we do not give the details here since this fact will be a consequence of Step 2.

\addtocontents{toc}{\SkipTocEntry} 
\subsection*{Step 2: Global uniqueness.}
Let $(M_n)_n$ be a sequence of matrices in ${\mathscr D}$ such that $M_n\underset{n\to\infty}{\longrightarrow} s\cdot\textrm{Id}$ and let $(\psi_{M_n})_n$ be a sequence of minimizers of $\left(I_{M_n}(\lambda)\right)_n$ which we can suppose, up to phase, strictly positive by~\cth{anisotrop_existence_minimiseur} and, up to a space translation (for each $M_n$), in $(\ker {\mathfrak L}_Q)^\perp$. Indeed, for any $\psi\in H^1(\R^3)$, let us define the continuous function $f(\tau):=\int\nabla Q(\cdot)\psi(\cdot-\tau)$ which is bounded, by the Cauchy-Schwarz inequality. Then $\int{f(\tau)\dd\tau}=\int{\psi(x)\int{\nabla Q(x-\tau)\dd \tau}\dd x}=0$ since $\int\nabla Q=0$. Thus, $f$ being continuous, there exists $\tau$ such that $f(\tau)=\int{\psi(x-\tau)\nabla Q(x) \dd x}=0$, i.e. $\psi(\cdot-\tau)\in (\ker {\mathfrak L}_Q)^\perp$ since $\ker {\mathfrak L}_Q=\vect\left\{\partial_{x_1}Q,\partial_{x_2}Q,\partial_{x_3}Q\right\}$.

By continuity of $(I_{M_n}(\lambda))_n$, given by~\clm{Conv_minimum_I_with_lambda_M}, $(\psi_{M_n})_n$ is a minimizing sequence of $I_{s\cdot\textrm{Id}}(\lambda)$. So, by~\cth{anisotrop_existence_minimiseur}, $(\psi_{M_n})_n$ strongly converges in $H^1(\R^3)$ to a minimizer of $I_{s\cdot\textrm{Id}}(\lambda)$, up to extraction of a subsequence. But, since the $\psi_{M_n}$ are positive and in $(\ker {\mathfrak L}_Q)^\perp$, they converge to a positive minimizer of $I_{s\cdot\textrm{Id}}(\lambda)$ in $(\ker {\mathfrak L}_Q)^\perp$ which is $Q$.

So, there exists $\epsilon'\leq\epsilon$ such that if $\norm{M-s\cdot\textrm{Id}}_\infty\leq\epsilon'$, then each $\psi_{M_n}$ verifies $G(\psi_{M_n},M_n)=0$, by definition of $\left(\psi_{M_n}\right)_n$, and $\norm{\psi_{M_n}-Q}_{H^1}\leq \delta$ i.e. verifies~\eqref{local_uniqueness_radial_implicit}. So the $\psi_{M_n}$ are unique (up to phases and spaces translation). Which concludes the proof of~\cth{Weakly_anisotr_uniqueness_minim}.

Moreover, we now know that, in fact, the masses $\norm{\psi(M_n)}_2^2$ of the unique $\psi(M_n)$ found in the local result were in fact all equal to $\lambda$. We also proved incidentally that our choice of translation to obtain $(\psi_{M_n})_n\subset (\ker {\mathfrak L}_Q)^\perp$ was, in fact, unique.
\end{proof}

\section{Rearrangements and symmetries}\label{section_rearr_prop}
The goal of this section is to prove that minimizers are symmetric and strictly decreasing in the directions along which $V$ is decreasing,  without assuming that $V$ is close to the isotropic case as we did in the previous section. More precisely, we will consider here the general anisotropic case $m_3\leq m_2\leq m_1$ (resp. $s_3\leq s_2\leq s_1$) and, in particular, the two cylindrical cases $m_3= m_2< m_1$ (resp. $s_3= s_2< s_1$) and $m_3< m_2= m_1$ (resp. $s_3< s_2= s_1$). Our main result in this section is~\cth{sym_minimiseurs_sous_condition} below. As a preparation, we first give conditions for $V$ to be its own Steiner symmetrization.

As in~\cite{Capriani-12}, for $f$ defined on $\R^n=\vect\{e_1,\dots,e_n\}$, we denote:
\begin{itemize}
\item by $f^*$ its \emph{Schwarz symmetrization}, for $n\geq1$;
\item by $\st_{i_1,\dots,i_k}(f)$ its \emph{Steiner symmetrization (in codimension $k$) with respect to the subspace spanned by $e_{i_1},\dots,e_{i_k}$},  for $n\geq2$ and $1\leq k<n$.
\end{itemize} Let us remark that the Steiner symmetrization $\st_{i_1,\dots,i_k}(f)$ of $f$ is the Schwarz symmetrization of the function $(x_{i_1}, \cdots, x_{i_k})\mapsto f(x_1, \cdots, x_n)$.

\begin{prop}[Criterion for $V$ to be its own Steiner symmetrization]\label{cond_sym_V_epsiM}\hspace{0.1cm}

\begin{enumerate}
	\item Let $V$ be given by~\eqref{V_final_definition}, with $0< m_3\leq m_2\leq m_1<1$. Then $V=\emph{\st}_1(V)$ (thus $V$ is $e_1$-symmetric strictly decreasing). Moreover, for $k\in\{2,3\}$, $V=\emph{\st}_k(V)$ (thus $V$ is $e_k$-symmetric strictly decreasing) if and only if
	\refstepcounter{equation}\label{cond_sym_V_M_default}
			\begin{equation*}\setcounter{equation}{1}\tag{\theequation$_k$}\label{cond_sym_V_M_k}
				m_1^3 \leq m_k^2.
			\end{equation*} Moreover,
	\begin{enumerate}[label=\roman*.]
		\item  if $m_3< m_2= m_1$, then $V=\emph{\st}_{1,2}(V)$. Thus $V$ is $(e_1,e_2)$-radial strictly decreasing.
		\item  if $m_3= m_2<m_1$, then $V=\emph{\st}_{2,3}(V)$ --- thus $V$ is $(e_2,e_3)$-radial strictly decreasing --- if and only if
			\begin{equation}\label{cond_sym_V_M_2_3}
				m_1^3\leq m_2^2 = m_3^2;
			\end{equation}
	\end{enumerate}
	\item Let $V$ be given by~\eqref{V_simplified_final_definition}, with $0\leq s_3\leq s_2\leq s_1<1$. Then $V=\emph{\st}_k(V)$ (thus $V$ is $e_k$-symmetric strictly decreasing) for $k=1,2,3$. Moreover,
	\begin{enumerate}[label=\roman*.]
		\item if $s_3< s_2=s_1$, then $V=\emph{\st}_{1,2}(V)$. Thus $V$ is $(e_1,e_2)$-radial strictly decreasing;
		\item if $s_3=s_2<s_1$, then $V=\emph{\st}_{2,3}(V)$. Thus $V$ is $(e_2,e_3)$-radial strictly decreasing.
	\end{enumerate}
\end{enumerate}
\end{prop}

\begin{proof}[Proof of~\cpr{cond_sym_V_epsiM}]Suppose $V$ is given by~\eqref{V_final_definition}, then it obviously has the claimed properties of symmetry and, moreover, the cylindrical ones in cases \emph{i.} and \emph{ii.}. So the proof that $V$ is equal to its symmetrization is reduced to the proof of decreasing properties.

For any $x\neq0$ and $k=1,2,3$, we have
\begin{equation}\label{cond_sym_V_epsiM_computation}
\partial_{|x_k|}V(x_1,x_2,x_3)=\frac{ m_k^{-2}|x_k|}{( m_1^{-2}x_1^2+ m_2^{-2}x_2^2+ m_3^{-2}x_3^2)^{3/2}}-\frac{|x_k|}{(x_1^2+x_2^2+x_3^2)^{3/2}}.
\end{equation} Thus, $V=\st_k(V)$ and $V$ is radially decreasing with respect to $x_k$ if and only if $$
0 \leq ( m_1^{-2}- m_k^{-4/3})x_1^2+( m_2^{-2}- m_k^{-4/3})x_2^2+( m_3^{-2}- m_k^{-4/3})x_3^2 \quad\textrm{ a.e. on } \R^3$$ which is equivalent to $m_1 \leq m_k^{2/3}$. Consequently, $V=\st_1(V)$ always holds.

If $m_3= m_2<m_1$, denoting $u=|(x_2,x_3)|$, and computing $\partial_uV$, we obtain that $V=\st_{2,3}(V)$ if and only if $m_1\leq m_2^{2/3} = m_3^{2/3}$, in which case $V$ is $(e_2,e_3)$-radial decreasing.

If $m_3< m_2= m_1$, denoting $u=|(x_1,x_2)|$, and computing $\partial_uV$, we obtain that $V=\st_{1,2}(V)$ if and only if $m_3 \leq  m_2^{2/3}=m_1^{2/3}$, which always holds thus $V$ is $(e_1,e_2)$-radial decreasing.

We now need to prove the strict monotonicity. Thanks to \eqref{cond_sym_V_epsiM_computation}, $\nabla V=0$ holds only on measure-zero sets (note that we use the computation but do not use any condition on $m_1$, $m_2$ and $m_3$ except that they are strictly less than $1$). Thus $\left|\{V=t\} \right|~=~0$ for any $t \in \R_+$ and then $\left|\{V^*=t\} \right|=0$ for any $t \in \R_+$. Hence $V^* \textrm{ is radially strictly decreasing}$. Same results of strict decreasing hold for Steiner symmetrizations since, as noted before, a Steiner symmetrization is a Schwarz symmetrization on a subspace.

The proof for $V$ given by~\eqref{V_simplified_final_definition} is very similar and easier.
\end{proof}

We now state our main result about the symmetries of minimizers.
\begin{thm}[Symmetries of minimizers]\label{sym_minimiseurs_sous_condition}Let $\lambda>0$.

\begin{enumerate}	
	\item Let $V$ be given by~\eqref{V_final_definition} and $\psi_M\geq0$ be a minimizer of $I_M(\lambda)$. Then, up to a space translation, $\psi_M$ is $e_1$-symmetric strictly decreasing. If $m_1^3 \leq m_2^2$ as in $(\ref{cond_sym_V_M_default}_2)$, then $\psi_M$ is also $e_2$-symmetric strictly decreasing. Finally, if $m_1^3 \leq m_3^2$ as in $(\ref{cond_sym_V_M_default}_3)$, then $\psi_M$ is additionally $e_3$-symmetric strictly decreasing. Moreover,
		\begin{enumerate}[label=\roman*.]
			\item if $m_3<m_2=m_1$, then $\psi_M$ is \emph{cylindrical strictly decreasing} with axis~$e_3$. Meaning that $\psi_M$ is $(e_1,e_2)$-radial strictly decreasing. If additionally $(\ref{cond_sym_V_M_default}_3)$ holds, then $\psi_M$ is \emph{cylindrical-even strictly decreasing} with axis~$e_3$. This means that $\psi_M$ is \emph{cylindrical strictly decreasing} with axis~$e_3$ and $e_3$-symmetric strictly decreasing;
			\item if $m_3=m_2<m_1$ and $m_1^3\leq m_2^2 = m_3^2$ as in~\eqref{cond_sym_V_M_2_3}, then $\psi_M$ is \emph{cylindrical-even strictly decreasing} with axis~$e_1$.
		\end{enumerate}
	\item Let $V$ be given by~\eqref{V_simplified_final_definition} and $\psi_S\geq0$ be a minimizer of $I_S(\lambda)$. Then, up to a space translation, $\psi_S$ is $e_k$-symmetric strictly decreasing for $k=1,2,3$. Moreover,
		\begin{enumerate}[label=\roman*.]
			\item if $s_3< s_2=s_1$, then $\psi_S$ is \emph{cylindrical-even strictly decreasing} with axis~$e_3$;
			\item If $s_3=s_2<s_1$, then $\psi_S$ is \emph{cylindrical-even strictly decreasing} with axis~$e_1$.
		\end{enumerate}
\end{enumerate}
\end{thm}

To prove the symmetry properties of the minimizers, we need symmetrizations of a minimizer to be minimizers, which is proved in the following lemma.
\begin{lemme}\label{Corollaire_nrj_psi_egale_nrj_symmetrisee}
Suppose that $V$, given by~\eqref{V_final_definition} or by~\eqref{V_simplified_final_definition}, verifies one of the \emph{symmetric strictly decreasing} property (resp. \emph{radial strictly decreasing} property) described in~\cpr{cond_sym_V_epsiM}, and define $\psi^\emph{\st}$ the symmetrization of $\psi$ corresponding to this symmetric strictly decreasing property  of $V$.

If $\psi$ is a minimizer then $\psi^\emph{\st}$ too. Moreover the following equalities hold
\begin{enumerate}[label=\roman*.]
	\item $\norm{\nabla\psi}_2^2=\norm{\nabla\psi^\emph{\st}}_2^2$,
	\item $\pscal{|\psi|^2,|\psi|^2\star V}_2=\pscal{|\psi^\emph{\st}|^2,|\psi^\emph{\st}|^2\star V}_2$.
\end{enumerate}
\end{lemme}
\begin{proof}[Proof of~\clm{Corollaire_nrj_psi_egale_nrj_symmetrisee}]On one hand, since the symmetrization conserves the $L^2$ norm and $\psi$ is a minimizer, we have ${\mathscr E}(\psi)\leq{\mathscr E}(\psi^\st)$. On the other hand, given the Riesz inequality (see~\cite{Burchard-96}), the fact that the kinetic energy is decreasing under symmetrizations (see Theorem 2.1 in~\cite{Capriani-12}) and since $V=V^\st$ by~\cpr{cond_sym_V_epsiM}, we have ${\mathscr E}(\psi)\geq{\mathscr E}(\psi^\st)$. So finally $I(\lambda)={\mathscr E}(\psi)={\mathscr E}(\psi^\st)$. Consequently, given~\eqref{anisotrop_existence_minimiseur_norm_equal} in~\cth{anisotrop_existence_minimiseur} and that minimizers $\psi$ and $\psi^\st$ have the same Lagrange multiplier $\mu=-3\lambda^2 I(1)$, we immediately obtain both equalities.
\end{proof}

Using the analycity of minimizers (\clm{sol_analytique}) we can now prove the strict monotonicity of Steiner symmetrizations of minimizers.
\begin{lemme}\label{psi_symetrisation_strict_decr}
Let $\lambda>0$ and $\psi$ be a real minimizer of $I(\lambda)$ for $V$ given by~\eqref{V_final_definition} or by~\eqref{V_simplified_final_definition}, then $\psi^*$ is radially strictly decreasing. Moreover, for any permutation $\{i,j,k\}$ of $\{1,2,3\}$, we have
\begin{enumerate}[label=\roman*.]
	\item for any $x \in \vect\{e_j,e_k\}$, $\emph{\st}_i(\psi)(x,\cdot)$ is radially strictly decreasing,
	\item for any
 $x \in \vect\{e_i\}$, $\emph{\st}_{j,k}(\psi)(x,\cdot)$ is radially strictly decreasing.
\end{enumerate}
\end{lemme}
\begin{proof}[Proof of~\clm{psi_symetrisation_strict_decr}]
By~\cth{anisotrop_existence_minimiseur}, $\psi$ is a solution of~\eqref{Pekar_euler_lagrange_notnormalize} in $H^2(\R^3, \R)$ with a real Lagrange multiplier $\mu$. Then, by the following lemma (see the Appendix of~\cite{Ricaud-PhD} for the proof), $\psi$ is real analytic.
\begin{lemme}\label{sol_analytique}
Any $\psi \in H^2(\R^3, \R)$ solution of~\eqref{Pekar_euler_lagrange_notnormalize} for $\mu\in \R$ is analytic.
\end{lemme}

Thus $\left|\{\psi=t\} \right|~=~0$ for any $t \in \R_+$ and this is equivalent to $\left|\{\psi^*=t\} \right|=0$ for any $t \in \R_+$. Hence $\psi^* \textrm{ is radially strictly decreasing}$.

Given that for any $1\leq k<3$ and any $x \in \R^{3-k}$, $\psi(x,\cdot)$ is analytic and since a Steiner symmetrization is a Schwarz symmetrization, we obtain $ii.$ and $iii.$ by the same reasoning to $\psi(x,\cdot)$.
\end{proof}

Finally, to prove our~\cth{sym_minimiseurs_sous_condition} on the symmetries of minimizers, we need a result on the case of equality in Riesz' inequality for Steiner's symmetrizations.
We emphasize that different Steiner symmetrizations do not commute in general. However, if the Steiner symmetrizations are made with respect to the vectors of an orthogonal basis then the radial strictly decreasing properties are preserved.

For shortness, we write $u^{\st_k}:=\st_{k}(u)$ and, in cylindrical cases, $u^{\st_{1,2}}:=\st_{1,2}(u)$ and $u^{\st_{2,3}}:=\st_{2,3}(u)$.

\begin{prop}[Steiner symmetrization: case of equality for $g$ strictly decreasing]\label{equality_Riesz}
Let $f, g, h$ be three measurable functions on $\R^3$ such that $g>0$ and $f,h\geq0$ where $0\neq f\in L^p(\R^3)$, with $1\leq p\leq+\infty$, and $0\neq h\in L^q(\R^3)$, with $1\leq q\leq+\infty$. Define $$J(f,g,h)=\frac{1}{2}\displaystyle{\int_{\R^3}}{\displaystyle{\int_{\R^3}}{f(x)g(x-y)h(y) \dd x} \dd y}\leq\infty.$$
\begin{enumerate}[leftmargin=0.6cm]
	\item Let $(i,j,k)$ be a permutation of $(1,2,3)$ and $J\left(f^{\emph{\st}_i},g,h^{\emph{\st}_i}\right)<\infty$. If for any $(x_j,x_k)\in \R^2$ the functions $g$, $f^{\emph{\st}_i}$ and $h^{\emph{\st}_i}$ are all strictly decreasing with respect to $|x_i|$, then $$J(f,g,h)=J\left(f^{\emph{\st}_i},g,h^{\emph{\st}_i}\right) \Leftrightarrow \exists\, a \in \R^3,
	\left\{
		\begin{aligned}
			f&=f^{\emph{\st}_i}(\cdot-a),\\
			h&=h^{\emph{\st}_i}(\cdot-a),
		\end{aligned}
	\right. \;\;\textrm{ a.e. on } \R^3.$$
	\item Let $(i,j,k)$ be a permutation of $(1,2,3)$ and $J\left(f^{\emph{\st}_{j,k}},g,h^{\emph{\st}_{j,k}}\right)<\infty$. If for any $x_i\in \R$ the functions $g$, $f^{\emph{\st}_{j,k}}$ and $h^{\emph{\st}_{j,k}}$ are all radially strictly decreasing with respect to $(x_j,x_k)$, then $$J(f,g,h)=J\left(f^{\emph{\st}_{j,k}},g,h^{\emph{\st}_{j,k}}\right) \Leftrightarrow \exists\, a \in \R^3,
	\left\{
		\begin{aligned}
			f&=f^{\emph{\st}_{j,k}}(\cdot-a),\\
			h&=h^{\emph{\st}_{j,k}}(\cdot-a).
		\end{aligned}
	\right. \;\;\textrm{ a.e. on } \R^3.$$
	\item Let $\emph{\st}$ and $\emph{\st}'$ be two Steiner symmetrizations, acting on two orthogonal directions, $T=\emph{\st}'\circ \emph{\st}$ and $J\left(f^T,g,h^T\right)<\infty$. If the functions $g$, $f^{\emph{\st}}$, $h^{\emph{\st}}$ are all radially strictly decreasing in the direction (or the plane) of $\emph{\st}$, and $g$, $f^{\emph{\st}'}$ and $h^{\emph{\st}'}$ are all radially strictly decreasing in the direction (or the plane) of $\emph{\st}'$, then $$J(f,g,h)=J\left(f^T,g,h^T\right) \Leftrightarrow \exists\, a \in \R^3,
	\left\{
		\begin{aligned}
			f&=f^T(\cdot-a),\\
			h&=h^T(\cdot-a).
		\end{aligned}
	\right. \;\;\textrm{ a.e. on } \R^3.$$
\end{enumerate}
\end{prop}
\begin{proof}[Proof of~\cpr{equality_Riesz}]
The implications $\Leftarrow$ all follow from a simple changes of variable. We show the implications $\Rightarrow$ and start with \emph{(1)}. Define, for any permutation $(i,j,k)$ of $(1,2,3)$ and any $(x_j,x'_j,x_k,x'_k)\in\R^4$, the functions 
$$J_i(f,g,h)(x_j,x'_j,x_k,x'_k)=\frac{1}{2}\displaystyle{\int_{\R}}{\displaystyle{\int_{\R}}{f(X)g(X-X')h(X') \dd x_i} \dd x'_i},$$  where $X=(x_1,x_2,x_3)$ and $X'=(x_1',x_2',x_3')$. We claim that for almost all $(x_j,x'_j,x_k,x'_k)\in\R^4$, we have $$J_i(f,g,h)(x_j,x'_j,x_k,x'_k)=J_i(f^{\st_i},g,h^{\st_i})(x_j,x'_j,x_k,x'_k).$$ Indeed, assume that there exists a non-zero measure set $E \subset \R^2\times\R^2$ such that $J_i(f,g,h)(y,y')\neq J_i(f^{\st_i},g,h^{\st_i})(y,y')$ for any $(y,y')\in E$. Thus, by Riesz inequality on $\R$, $J_i(f,g,h)<J_i(f^{\st_i},g,h^{\st_i})$ necessarily holds on $E$, since $g=g^{\st_i}$, and consequently $J(f,g,h)<J(f^{\st_i},g,h^{\st_i})$, reaching a contradiction.

We now use the following result of Lieb~\cite{Lieb-77}: 
\begin{lemme}[\texorpdfstring{\cite[Lemma 3]{Lieb-77}}{}: Case of equality in Riesz' inequality for $g$ strictly decreasing]\label{Ineq_Riesz_case_equality_three_strict_decreas_rearrang_and_g_symm}
Suppose $g$ is a positive spherically symmetric strictly decreasing function on $\R^n$, $f\in L^p(\R^n)$ and $h\in L^q(\R^n)$ are two nonnegative functions, with  $p,q\in[1;+\infty]$, such that $J(f^*,g,h^*)<\infty$. Then $$J(f,g,h)=J(f^*,g,h^*)\, \Rightarrow\, \exists\; a \in \R^n,\; f=f^*(\cdot-a) \textrm{ and } h=h^*(\cdot-a)\;\; \textrm{ a.e.}.$$
\end{lemme}

Thus, for almost all $(y,y')\in \R^2\times\R^2$, there exists  $a_i(y,y') \in \R$ such that $f(y,x_i)=f^{\st_i}\left(y, x_i - a_i(y,y')\right)$ and $h(y',x_i)=h^{\st_i}\left(y',x_i - a_i(y,y')\right)$,  for almost all $x_i\in \R$. Using now the assumed strict monotonicity of $f^{\st_i}(y,\cdot)$ and $h^{\st_i}(y',\cdot)$, it follows that $a_i$ does not depend on $(y,y')$, and \emph{(1)} is proved.

The case \emph{(2)} is very similar, defining this time $$J_{j,k}(f,g,h)(x_i,x_i')=\frac12\pscal{f(\cdot,x_i), g(\cdot,x_i-x_i')\star h(\cdot,x_i')}_{L^2(\R^2)},\, \forall(x_i,x_i')\in\R^2.$$

We now prove \emph{(3)}. Let $\st$ be one of the Steiner's symmetrization described \emph{(1)} and \emph{(2)} and the same for $\st'$. We claim that $$J_\st(f,g,h)=J_\st(f^{\st},g,h^{\st}) \textrm{ and } J_{\st'}(f,g,h)=J_{\st'}(f^{\st'},g,h^{\st'}), \; a.e..$$ Indeed, Riesz inequality gives $J(f,g,h)\leq J(f^{\st},g,h^{\st})\leq J(f^T,g,h^T)$. Since first and third terms are equal, the three of them are. From the first equality, there exists $a\in\R^\ell$ ($\ell=1,2$) such that $f=f^\st(\cdot-a,\cdot)$ and $h=h^\st(\cdot-a,\cdot)$. Then, since $\st$ and $\st'$ act on orthogonal directions, we have $$J(f^T,g,h^T)=J\left(f^{\st'}(\cdot+a,\cdot),g,h^{\st'}(\cdot+a,\cdot)\right)=J(f^{\st'},g,h^{\st'})$$ and so the second claim holds true too. Then, for almost every $y:=(x,z)\in \R^3$, we have
\begin{equation*}
\left\{
	\begin{aligned}
		f^{T}\left(y-(a', a)\right)&=\left(f^{\st}(x-a,\cdot)\right)^{\st'}(z-a')=f^{\st'}(x, z-a')=f(x,z)=f(y),\\
		h^{T}\left(y-(a', a)\right)&=\left(h^{\st}(x-a,\cdot)\right)^{\st'}(z-a')=h^{\st'}(x, z-a')=h(x,z)=h(y).
	\end{aligned}
\right.
\end{equation*}
\end{proof}
We now have all the ingredients to prove~\cth{sym_minimiseurs_sous_condition}.
\begin{proof}[Proof of~\cth{sym_minimiseurs_sous_condition}]Let $\psi$ be a minimizer and $\psi^\st$ one (or a composition) of its Steiner symmetrizations with a direction (or a plane) for which $V=V^\st$.

We take $f=h=|\psi|^2\in$ and $g=V$. So we have $f=h>0$ (thanks to~\cth{anisotrop_existence_minimiseur}), $g>0$ (thanks to~\eqref{aniso_propri_V}) and $J(f^\st,V,f^\st)$ finite. Indeed by the Hardy-Littlewood-Sobolev inequality and~\eqref{aniso_propri_V}, $J(f^\st,V,f^\st)\lesssim \norm{f^\st}^2_{6/5}=\norm{f}^2_{6/5}<+\infty$ since $f\in H^1(\R^3)$. Moreover, the assumption on the $m_k$'s gives that $g=g^\st$ is radially strictly decreasing by~\cpr{cond_sym_V_epsiM}, and the strict monotonicity of $f^\st=h^\st$ is obtained by~\clm{psi_symetrisation_strict_decr}.

Finally, by \clm{Corollaire_nrj_psi_egale_nrj_symmetrisee}, $\psi^\st$ is a minimizer and $$J(|\psi|^2,V,|\psi|^2)=J\left((|\psi|^2)^\st,(V)^\st,(|\psi|^2)^\st\right)=J\left((|\psi|^2)^\st,V,(|\psi|^2)^\st\right).$$ By~\cpr{equality_Riesz}, there exists $a$ such that $|\psi|^2=(|\psi|^2)^\st(\cdot-a)=(|\psi|^\st)^2(\cdot-a)$ holds a.e. thus $\psi=\psi^\st(\cdot-a)$ since $\psi\geq0$. This concludes the proof of~\cth{sym_minimiseurs_sous_condition}.
\end{proof}

\section{Study of the linearized operator}\label{chapter_linearized_oper}
In this section we study the linearized operator ${\mathfrak L}_Q$, on $L^2(\R^3)$ with domain $H^2(\R^3)$, associated with the Euler-Lagrange equation $-\Delta Q + Q - (|Q|^2\star V)Q=0$~\eqref{Pekar_euler_lagrange_notnormalize}, which is given by
\begin{equation}\label{op_linearise_adimensionne_symm}
\boxed{{\mathfrak L}_Q\xi=-\Delta\xi +\xi -\left(V\star|Q|^2\right)\xi -2Q\left(V\star(Q\xi)\right),}
\end{equation} and we give partial characterization of its kernel. We first consider the true model~\eqref{V_final_definition} for which, following the scheme in~\cite{Lenzmann-09}, we will use a Perron--Frobenius argument on subspaces adapted to the symmetries of the problem. The main difficulty will stand in dealing with the non-local operator $Q\left(V\star(Q\xi)\right)$ and, in particular, with proving that this operator is positivity improving. The fundamental use of Newton's theorem in the proof of this property in the isotropic case does not work here, therefore we need a new argument. Our proof will rely on the conditions~\eqref{cond_sym_V_M_k}'s for which $V$ is $e_k$-symmetric strictly decreasing for each $k$ (see~\cpr{cond_sym_V_epsiM}). Then we discuss in a similar way the cylindrical case for the simplified model~\eqref{V_simplified_final_definition}, which will need another argument.

\subsection{The linearized operator in the symmetric decreasing case}\label{chapter_linearized_oper_symm}
We consider the general case for $V$, given by~\eqref{V_final_definition}, verifying the three conditions~\eqref{cond_sym_V_M_k}, for $k=1,2,3$, and define the subspaces of $L^2(\R^3)$
\begin{equation}\label{decompo_L2_definition_ensemble_symm}
L^2_{\tau_x,\tau_y,\tau_z}:=
\left\{f\in L^2(\R^3)\, \left|
	\begin{aligned}
		f(-x,y,z)&= \tau_x f(x,y,z),\\
		f(x,-y,z)&= \tau_y f(x,y,z),\\
		f(x,y,-z)&= \tau_z f(x,y,z)
	\end{aligned}\right.
\right\},
\end{equation} obtained by choosing $\tau_x,\tau_y,\tau_z\in\{\pm1\}$. We prove the following theorem which is basically saying that the kernel of the linearized operator around solutions is reduced to the kernel on functions that are even in all three directions.

\begin{thm}\label{THM_Ker_L_symm}
Let $V$, be given by~\eqref{V_final_definition}, verifying~\eqref{cond_sym_V_M_k}, for all $k$, and let $Q$ be a positive and symmetric strictly decreasing (with respect to each $e_k$ separately) solution of~\eqref{Pekar_euler_lagrange_notnormalize}. Then
\begin{equation}\label{Ker_of_L_symm}
\ker {\mathfrak L}_Q=\vect\left\{\partial_x Q, \partial_y Q, \partial_z Q\right\} \bigoplus \ker \left({\mathfrak L}_Q\right)_{|L^2_{+,+,+}}.
\end{equation}
For instance, $Q$ could be a minimizer for $I_M(\lambda)$.
\end{thm} The proof of this result is inspired by Lenzmann's proof in~\cite{Lenzmann-09} of the fundamental similar result for the linearized operator in the radial case which corresponds to $m_1=m_2=m_3$. In that case, Lenzmann proved that $\ker \left({\mathfrak L}_Q\right)_{|L^2_{+,+,+}}=\{0\}$. Note that by the result of Section~\ref{uniqueness_weakly_anisotropic}, we know that this is still true in the weakly anisotropic case. Moreover, a theorem similar to~\cth{THM_Ker_L_symm} holds true for the simplified model~\eqref{V_simplified_final_definition} (with no conditions on the matrix $S$) but we do not state it here for shortness.

The rest of this Section~\ref{chapter_linearized_oper_symm} being dedicated to the proof of the theorem, let $V$ and $Q$ verify the assumptions of~\cth{THM_Ker_L_symm} for the entire Section~\ref{chapter_linearized_oper_symm}.

\subsubsection{Direct sum decomposition} First, one can easily verify that ${\mathfrak L}_Q$ stabilizes the spaces $L^2_{\tau_x,\tau_y,\tau_z}$. Let us then introduce the direct sum decomposition
\begin{equation*}\label{decompo_L2_symm}
L^2(\R^3)=L^2_{x-}\oplus L^2_{x+}=L^2_{y-}\oplus L^2_{y+}=L^2_{z-}\oplus L^2_{z+}
\end{equation*} where
\begin{equation*}\label{decompo_L2_symm_def}
\left\{
\begin{aligned}
L^2_{x-}&:=\bigoplus\limits_{\tau_y,\tau_z=\pm}{L^2_{-,\tau_y,\tau_z}}, \quad L^2_{x+}:=\bigoplus\limits_{\tau_y,\tau_z=\pm}{L^2_{+,\tau_y,\tau_z}}\\
L^2_{y-}&:=\bigoplus\limits_{\tau_x,\tau_z=\pm}{L^2_{\tau_x,-,\tau_z}}, \quad L^2_{y+}:=\bigoplus\limits_{\tau_x,\tau_z=\pm}{L^2_{\tau_x,+,\tau_z}}\\
L^2_{z-}&:=\bigoplus\limits_{\tau_x,\tau_y=\pm}{L^2_{\tau_x,\tau_y,-}}, \quad L^2_{z+}:=\bigoplus\limits_{\tau_x,\tau_y=\pm}{L^2_{\tau_x,\tau_y,+}}.
\end{aligned}\right.
\end{equation*}

We claim that those spaces --- with corresponding projectors $P^{x-}$, $P^{x+}$, $P^{y-}$, $P^{y+}$, $P^{z-}$ and $P^{z+}$ --- reduce the linerized operator ${\mathfrak L}_Q$ (see~\cite{Teschl-09} for a definition of reduction), where
\begin{align*}
P^{x\pm}\psi(r,\phi,z)&=\frac{\psi(x,y,z)\pm\psi(-x,y,z)}2
\end{align*} and similarly for the other projections. The reduction property is straightforward for $-\Delta +1 -\left(V\star|Q|^2\right)$. Moreover, since $Q$ is even in $x$, we have
\begin{align*}
V\star(QP^{x\pm}\psi)&=\frac{V\star(Q\phi)\pm V\star(Q\phi(-\cdot,\cdot,\cdot))}2\\
&=\frac{V\star(Q\phi)\pm[V\star(Q\phi)](-\cdot,\cdot,\cdot)}2=P^{x\pm}[V\star(Q\psi)]
\end{align*} and, $Q$ being also even in $y$ and in $z$, we obtain the result for the other projections. Thus  we can apply~\cite[Lemma 2.24]{Teschl-09} which gives us that 
\begin{equation*}\label{decompo_L2_symm_operator_L}
{\mathfrak L}_Q={\mathfrak L}_Q^{x-}\oplus {\mathfrak L}_Q^{x+}={\mathfrak L}_Q^{y-}\oplus {\mathfrak L}_Q^{y+}={\mathfrak L}_Q^{z-}\oplus {\mathfrak L}_Q^{z+},
\end{equation*} with the six operators being self-adjoint operators on the corresponding $L^2(\R^3)$ spaces with domains $P^w H^2(\R^3)$, and $w\in\{x-,x+,y-,y+,z-,z+\}$. Note that $P^{x-}H^2(\R^3)=H^2_{x-}(\R^3):=H^2(\R^3)\cap L^2_{x-}(\R^3)$ and similarly for $P^{y-}$ and $P^{z-}$.

Let us then redefine for now on the operator ${\mathfrak L}_Q^{x-}$ (resp. ${\mathfrak L}_Q^{y-}$ and ${\mathfrak L}_Q^{z-}$) by restricting it to $x$-odd (resp. $y$-odd and $z$-odd) functions through the isomorphic identifications $L^2_{x-}(\R_+^*\times\R^2)\approx L^2_{x-}(\R^3)$ and $H^2_{x-}(\R_+^*\times\R^2)\approx H^2_{x-}(\R^3)$. Thus, ${\mathfrak L}_Q^{x-}$, as an operator on $L^2_{x-}(\R_+^*\times\R^2)$ with domain $H^2_{x-}(\R_+^*\times\R^2)$, can be written
\begin{equation*}\label{def_L_moins}
{\mathfrak L}_Q^{x-}=-\Delta+1+\Phi_{(-)}+W_{(-)}
\end{equation*} where the strictly negative multiplication local operator, on $\R_+^*\times\R^2$, is
\begin{align*}
\Phi_{(-)}(x,Y)&=-\left(V\star| Q|^2\right)(x,Y)\\
	&=-\displaystyle{\int_{\R_+^*\times\R^2}}{\!\!Q^2(x',Y') [V(x-x',Y-Y')+V(x+x',Y-Y')] \dd Y' \dd x'}
\end{align*} and the non-local term $W_{(-)}$, on $\R_+^*\times\R^2$, is
\begin{multline*}
(W_{(-)}f)(x,Y)=-2 Q(x,Y) \times\\
	\times\displaystyle{\int_{\R_+^*\times\R^2}}{Q(x',Y')f(x',Y')[V(x-x',Y-Y')-V(x+x',Y-Y')] \dd Y' \dd x'}.
\end{multline*}
The same properties hold for ${\mathfrak L}_Q^{y-}$ and ${\mathfrak L}_Q^{z-}$.

The key fact to deal with the non-local operator, in order to adapt Lenzmann's proof to anisotropic case, is the positivity improving property of $-W_{(-)}$.
\begin{lemme}\label{W_pos_improv}
The operator $-W_{(-)}$ is positivity improving on $L^2_{x-}(\R_+^*\times\R^2)$.
\end{lemme}
\begin{proof}[Proof of~\clm{W_pos_improv}]
Since $X\mapsto V(X,Y)$ is $|X|$-strictly decreasing, due to conditions~\eqref{cond_sym_V_M_k}, and $x+x'>|x-x'|$ on $(\R_+^*)^2$, we obtain, for $x,x'>0$ and $(Y,Y')\in\left(\R^2\right)^2$, that $V(x-x',Y-Y')-V(x+x',Y-Y')>0$. Moreover $Q>0$. Thus, $-W_{(-)}$ is positivity improving on $L^2_{x-}(\R_+^*\times\R^2)$.
\end{proof}

\subsubsection{Perron--Frobenius property}\label{Perron-Frobenius_property}
We can now prove that the three operators ${\mathfrak L}_Q^{x-}$, ${\mathfrak L}_Q^{y-}$ and ${\mathfrak L}_Q^{z-}$ verify a Perron--Frobenius property.
\begin{prop}[Perron--Frobenius properties]\label{symm_Perron_fro_ppty}The operators ${\mathfrak L}_Q^{x-}$, ${\mathfrak L}_Q^{y-}$ and ${\mathfrak L}_Q^{z-}$ are self-adjoint on $L^2_-(\R_+^*\times\R^2)$ with domain $H^2_-(\R_+^*\times\R^2)$ and bounded below.

Moreover they have the Perron--Frobenius property: if $\lambda^{x-}_0$ (resp. $\lambda^{y-}_0$ and $\lambda^{z-}_0$) denotes the lowest eigenvalue of ${\mathfrak L}_Q^{x-}$ (resp. ${\mathfrak L}_Q^{y-}$ and ${\mathfrak L}_Q^{z-}$), then $\lambda^{x-}_0$ (resp. $\lambda^{y-}_0$ and $\lambda^{z-}_0$) is simple and the corresponding eigenfunction $\psi^{x-}_0$ (resp. $\psi^{y-}_0$ and $\psi^{z-}_0$) is strictly positive.
\end{prop}
\begin{proof}[Proof of \cpr{symm_Perron_fro_ppty}] We follow the structure of the proof of \cite[Lemma 8]{Lenzmann-09}. Moreover, we only write the proof for ${\mathfrak L}_Q^{x-}$ which we denote ${\mathfrak L}_Q^{-}$ for simplicity. The argument is the same for the other directions.

\medskip

\noindent\textbf{Self-adjointness.} We have $Q \in H^2(\R^3)\subset C^0(\R^3)\cap L^2(\R^3)\cap L^\infty(\R^3)$ and, by \eqref{aniso_propri_V}, $V\star \left| Q\right|^2$ is in $L^4(\R^3)\cap L^\infty(\R^3)$ since $V=V_2+V_4\in L^2(\R^3)+ L^4(\R^3)$. Defining, for any $f\in L^2_-(\R_+^*\times\R^2)$, $\tilde{f}\in L^2_-(\R^3)$ by $f(\cdot,z)=\tilde{f}(\cdot,z)$ for $z\geq0$, we have $2\pscal{f,g}_{L^2_-(\R_+^*\times\R^2)}={\langle \tilde{f},\tilde{g} \rangle}_{L^2_-(\R^3)}$ and so $\Phi_{(-)}+1$ is bounded on $L_-^2(\R_+^*\times\R^2)$. Moreover, by Young inequalities, for any $\xi\in L_-^2(\R_+^*\times\R^2)$,
\begin{align*}
&\normSM{V\star (Q\tilde{\xi})}_{L^\infty}\leq\left(\norm{V_4}_{L^4}\norm{ Q}_{L^4} + \norm{V_2}_{L^2}\norm{ Q}_{L^\infty}\right)\normSM{\tilde{\xi}}_{L^2}
\end{align*} holds. Thus, for $p\in[2,\infty]$, we have $$\norm{W_{(-)}\xi}_{L^p(\R_+^*\times\R^2)}\leq2 \norm{Q}_{L^p(\R_+^*\times\R^2)}\normSM{V\star( Q\tilde{\xi})}_{L^\infty(\R_+^*\times\R^2)}\leq \norm{Q}_{L^p}\normSM{V\star( Q\tilde{\xi})}_{L^\infty}$$ and $W_{(-)}\xi\in L^2_-(\R_+^*\times\R^2)\cap L^\infty(\R_+^*\times\R^2)$. Finally, $1+\Phi_{(-)}+W_{(-)}$ and, thus, ${\mathfrak L}_Q^-$ are bounded below on $L_-^2(\R_+^*\times\R^2)$.

Finally, we deduce the self-adjointness of the operator ${\mathfrak L}_Q^-$ on $L^2_-(\R_+^*\times\R^2)$ with domain $H^2_-(\R_+^*\times\R^2)$ from the self-adjointness of the operator ${\mathfrak L}_Q^-$ on $L^2_-(\R^3)$ with domain $H^2_-(\R^3)$. 

\medskip

\noindent\textbf{Positivity improving.}
We know (see~\cite{LieLos-01} for example) that $$(-\Delta+\mu)^{-1}\xi(X)=\frac1{4\pi}\int_{\R^3}{\frac{e^{-\sqrt{\mu}|X-Y|}}{|X-Y|}\xi(Y) \dd Y}, \qquad \forall \mu>0,\; \forall \xi\in L^2(\R^3).$$ Consequently, for $\xi\in L_-^2(\R_+^*\times\R^2)$ and $(x,\tilde{x})\in\R_+^*\times\R^2$, we have $$(-\Delta+\mu)^{-1}\xi(x,\tilde{x})=\frac1{4\pi}\int_{\R_+^*\times\R^2}{\left[\frac{e^{-\sqrt{\mu}|(x-y,\tilde{x}-\tilde{y})|}}{|(x-y,\tilde{x}-\tilde{y})|} - \frac{e^{-\sqrt{\mu}|(x+y,\tilde{x}-\tilde{y})|}}{|(x+y,\tilde{x}-\tilde{y})|}\right]\xi(y,\tilde{y}) \dd y \dd \tilde{y}}.$$ Thus, with the same arguments as in the proof of~\clm{W_pos_improv}, $(-\Delta + \mu)^{-1}$ is positivity improving on $L^2_-(\R_+^*\times\R^2)$ for all $\mu>0$. Moreover, $-(\Phi_{(-)} + W_{(-)})$ is positivity improving on $L^2_-(\R_+^*\times\R^2)$ since $-\Phi_{(-)}$ is a positive multiplication operator and $-W_{(-)}$ is positivity improving by~\clm{W_pos_improv}. Then similarly to the proof of \cite[Lemma 8]{Lenzmann-09}, for $\mu\gg1$,we have $$\left({\mathfrak L}_Q^-+\mu\right)^{-1}=\left(-\Delta+\mu+1\right)^{-1}\cdot\left(1+(\Phi_{(-)} + W_{(-)})(-\Delta+\mu+1)^{-1}\right)^{-1}.$$ Since $(\Phi_{(-)} + W_{(-)})$ is bounded, we have $$\norm{(\Phi_{(-)} + W_{(-)})(-\Delta+\mu)^{-1}}_{L^2_-(\R_+^*\times\R^2)}<1,$$ for $\mu$ large enough. This implies, for $\mu\gg1$, by Neumann's expansion that $$\left({\mathfrak L}_Q^-+\mu\right)^{-1}=\left(-\Delta+\mu+1\right)^{-1}\sum\limits_{p=0}^\infty{\left[-(\Phi_{(-)} + W_{(-)})(-\Delta+\mu+1)^{-1}\right]^p},$$ which is consequently positivity improving on $L^2_-(\R_+^*\times\R^2)$ for $\mu\gg1$.

\medskip

\noindent\textbf{Conclusion.} We choose $\mu\gg1$ such that $({\mathfrak L}_Q^-+\mu)^{-1}$ is positivity improving and bounded. Then, by \cite[Thm XIII.43]{ReeSim4}, the largest eigenvalue $\sup\sigma(({\mathfrak L}_Q^-+\mu)^{-1})$ is simple and the associated eigenfunction $\psi^-_0\in L^2_-\left(\R_+^*\times\R^2\right)$ is strictly positive. Since, for any $\psi\in L^2_-\left(\R_+^*\times\R^2\right)$, having $\psi$ being an eigenfunction of ${{\mathfrak L}_Q^-}$ for the eigenvalue $\lambda$ is equivalent to having $\psi$ being an eigenfunction of $({\mathfrak L}_Q^-+\mu)^{-1}$ for the eigenvalue $(\lambda+\mu)^{-1}$, we have proved \cpr{symm_Perron_fro_ppty}.
\end{proof}

\subsubsection{Proof of \cth{THM_Ker_L_symm}}
Differentiating, with respect to $x$ the Euler-Lagrange equation $-\Delta Q + Q - (|Q|^2\star V)Q=0$~\eqref{Pekar_euler_lagrange_notnormalize}, we obtain ${\mathfrak L}_Q\partial_x Q\equiv0$. Moreover, $Q$ is positive \emph{symmetric strictly decreasing}, thus $\partial_x Q \in L^2_{x-}(\R^3)$, and this shows that ${\mathfrak L}_Q^{x-}\partial_x Q\equiv0$. Then, $Q>0$ being \emph{symmetric strictly decreasing}, $\partial_x Q<0$ on $\R_+^*\times\R^2$ and, by the Perron--Frobenius property, it is (up to sign) the unique eigenvector associated with the lowest eigenvalue of ${\mathfrak L}_Q^{x-}$, namely $\lambda^{x-}_0=0$. Since ${\mathfrak L}_Q^{x-}$ acts on $L^2_{x-}:=\bigoplus\limits_{\tau_y,\tau_z=\pm}{L^2_{-,\tau_y,\tau_z}}$, we obtain
\begin{equation*}
	\left\{\begin{aligned}
		\ker \left({\mathfrak L}_Q\right)_{|L^2_{-,+,+}(\R^3)}&=\vect\{\partial_x Q\};\\
		\ker \left({\mathfrak L}_Q\right)_{|L^2_{-,-,+}(\R^3)}&=\ker \left({\mathfrak L}_Q\right)_{|L^2_{-,+,-}(\R^3)}=\ker \left({\mathfrak L}_Q\right)_{|L^2_{-,-,-}(\R^3)}=\{0\}.
	\end{aligned}\right.
\end{equation*} This the exact same arguments for the two other directions we finally obtain that $$\ker {\mathfrak L}_Q=\vect\left\{\partial_x Q, \partial_y Q, \partial_z Q\right\} \bigoplus \ker \left({\mathfrak L}_Q\right)_{|L^2_{+,+,+}(\R^3)}.$$ Which concludes the proof of~\cth{THM_Ker_L_symm}.
\hfill\qedsymbol

\subsection{The linearized operator in the cylindrical case}\label{chapter_linearized_oper_cyl}
We now consider the case where the static dielectric matrix has exactly two identical eigenvalues. Obviously,~\cth{THM_Ker_L_symm} holds and it tells us that the kernel is reduced to the kernel on functions that are even in the $z$-direction and even in any direction of the plane orthogonal to $z$. However, this does not tell us that it is reduced to the kernel on cylindrical functions, which is what we are interested in. Indeed, instead of the kernel of ${\mathfrak L}_Q$ on $L^2_{+,+,+}(\R^3)$, we want the remaining term in the direct sum to be the kernel on $L^2_{\textrm{rad},+}(\R^3)$, namely the subset of cylindrical functions that are also even in the direction of their principal axis.

Unfortunately, our method fails to prove it for $V$ given by~\eqref{V_final_definition} since we are not able to prove a positivity improving property for the non local operator. Therefore, in this section, we will only consider the simplified model where $V$ is given by~\eqref{V_simplified_final_definition}.

We use the cylindrical coordinates $(r,z)$ where $e_z$ is the vector orthogonal to the plane of symmetry. Namely, $e_z=e_3$ if $s_3<s_2=s_1$ and $e_z=e_1$ if $s_3=s_2<s_1$. We then define the following subspaces
\begin{equation}\label{decompo_L2_cylin_definition_ensemble}
\begin{aligned}
L^2_\tau(\R^3)&:=\left\{f\in L^2(\R^3) \, | \,  f(x,y,- z)=\tau f(x,y,z)\right\}, &\textrm{ for } \tau=\pm \\
L^2_+(\R_+^*\times\R)&:=\left\{f\in L^2(\R_+^*\times\R,r\dd r\dd z) \, | \,  f(r,- z)= f(r,z)\right\} \\
L^2_{\textrm{rad},+}(\R^3)&:=L^2_+(\R_+^*\times\R)\otimes\mathds{1}.
\end{aligned}
\end{equation}

\begin{thm}\label{THM_Ker_L}

Let $V$ be given by~\eqref{V_simplified_final_definition} with $S$ having one eigenvalue of multiplicity $2$ and let $Q$ be a \emph{cylindrical-even decreasing} and positive solution of~\eqref{Pekar_euler_lagrange_notnormalize}. Then
\begin{equation}\label{Ker_of_L}
\ker {\mathfrak L}_Q=\vect\left\{\partial_x Q, \partial_y Q, \partial_z Q\right\} \oplus \ker \left({\mathfrak L}_Q\right)_{|L^2_{\emph{\textrm{rad}},+}(\R^3)}.
\end{equation}
For instance, $Q$ could be a minimizer for $I_S(\lambda)$.
\end{thm} Several parts of the proof of this theorem being identical to the ones in the proof of the~\cth{THM_Ker_L_symm}, we will only give the details for the parts that differ.

\subsubsection{Cylindrical decomposition}
Since $V$ is \emph{cylindrical-even strictly decreasing} by \cpr{cond_sym_V_epsiM} and since minimizers are \emph{cylindrical-even strictly decreasing} by \cpr{sym_minimiseurs_sous_condition}, ${\mathfrak L}_Q$ commutes with rotation in the plane of symmetry. Let us then introduce the direct sum decomposition
\begin{equation}\label{decompo_L2_cylin}
L^2(\R^3)=L^2_-(\R^3)\oplus\bigoplus\limits_{n\geq0,\sigma=\pm}{L^2_+(\R_+^*\times\R) \otimes Y^{\sigma}_{n}},
\end{equation} with
\begin{equation}\label{Def_Y_n}
\left\{\begin{aligned}
Y^+_{0}&\equiv(2\pi)^{-\frac12}; \; Y^-_{0}\equiv0 ;\\
Y^+_{n}&\equiv\pi^{-\frac12}\cos(n\cdot); \; Y^-_{n}\equiv\pi^{-\frac12}\sin(n\cdot), \;\; \textrm{ for } n\geq1.
\end{aligned}\right.
\end{equation} The operator ${\mathfrak L}_Q$ stabilizes $L^2_-(\R^3)$ and the spaces $L^2_+(\R_+^*\times\R) \otimes Y^{\sigma}_{n}$.

Let us immediately decompose the potential $V$ in order to give the fundamental property in the cylindrical case (\cpr{Cyl_strict_pos_v_n} below), which is what allows us to adapt the original work of Lenzmann, namely the strict positivity of each $z$-odd terms of the cylindrical decomposition of $V$. For any $\mathbf{r}=(r,\phi,z)\in\R^3$ and $\mathbf{r'}=(r',\phi',z')\in\R^3$, defining $\mathbf{\rho}:=(r-r',0,z-z')$ and $\theta:=\phi-\phi'$, we have, as soon as $(r',z')\neq(r,z)$:
\begin{align}\label{V_explicit}
\theta\mapsto V(\mathbf{r}-\mathbf{r'})&=\frac1{\sqrt{\left|(1-S)^{-1}\rho\right|^2+2(1-s_2)^{-2}rr'(1-\cos\theta)}}>0,
\end{align} which is in $C^\infty(\R)$, $2\pi$-periodic and even. Thus, for any $\mathbf{r}\neq\mathbf{r'}$,
\begin{align*}
V(\mathbf{r}-\mathbf{r'})&=\sum\limits_{n=0}^\infty{v_n(r,r',z-z')Y^+_{n}(\phi-\phi')},
\end{align*} with
\begin{align}\label{def_coeff_decomp}
v_n(r,r',z-z')=\displaystyle{\int_{-\pi}^\pi}{V(\mathbf{r}-\mathbf{r'})Y^+_{n}(\theta)\dd\theta }=2\displaystyle{\int_{0}^\pi}{V(\mathbf{r}-\mathbf{r'})Y^+_{n}(\theta)\dd\theta }.
\end{align}

\begin{prop}\label{Cyl_strict_pos_v_n}Let $V$ be given by~\eqref{V_simplified_final_definition}, the $Y^+_n$'s by~\eqref{Def_Y_n} and the $v_n$'s by~\eqref{def_coeff_decomp} for any $(n,r,r',z,z')\in\N\times\R_+^*\times\R_+^*\times\R\times\R$. Then $$v_n(r,r',z-z')>0, \quad\quad \forall (n,r,r',z,z')\in\N\times\R_+^*\times\R_+^*\times\R\times\R.$$
\end{prop}
\begin{proof}[Proof of~\cpr{Cyl_strict_pos_v_n}] Defining for $r,r'>0$,
\begin{align*}
m^\pm&:=\sqrt{\left(\frac{r+r'}{1-s_2}\right)^2+\left(\frac{Z}{1-s_z}\right)^2}\pm\sqrt{\left(\frac{r-r'}{1-s_2}\right)^2+\left(\frac{Z}{1-s_z}\right)^2}\\
	&=\max\limits_{\phi-\phi'}{\left|(1-S)^{-1}(\mathbf{r}-\mathbf{r'})\right| }\pm\min\limits_{\phi-\phi'}{\left|(1-S)^{-1}(\mathbf{r}-\mathbf{r'})\right| }>0,
\end{align*}
we note that $m^+>m^-$ and obtain
\begin{equation*}
V(\mathbf{r}-\mathbf{r'})=\frac{2}{m^+}\frac1{\sqrt{1-2\frac{m^-}{m^+}\cos\theta+\left(\frac{m^-}{m^+}\right)^2}}.
\end{equation*} We now give the explicit expansion of $(1-2t\cos\theta+t^2)^{-1/2}$ in the following lemma.
\begin{lemme}\label{Taylor_expansion_Legendre}
For $(0,1)\neq(\theta,t)\in \R\times[0,1]$, we have
\begin{equation}
\begin{aligned}
	\frac{1}{\sqrt{1-2t\cos\theta+t^2}}=&\sum\limits_{k=0}^\infty{\beta_{0,2k}t^{2k}} + \sum\limits_{n=1}^\infty{\sum\limits_{k=n}^\infty{\beta_{2n,2k} t^{2k}}Y^+_{2n}(\theta)}\\
	&+ \sum\limits_{n=0}^\infty{\sum\limits_{k=n}^\infty{\beta_{2n+1,2k+1} t^{2k+1}}Y^+_{2n+1}(\theta)}.
\end{aligned}
\end{equation}
with
\begin{equation*}
	\left\{
		\begin{aligned}
			\beta_{0,2k}&=\sqrt{2\pi}\frac{\binom{2k}{k}^2}{2^{4k}}>0, &\;\;\;0\leq k; \\
			\beta_{2n,2k}&=2\sqrt{\pi}\frac{\binom{2(k+n)}{k+n}\binom{2(k-n)}{k-n}}{2^{4k}}>0, &\;\;\;0< n\leq k; \\
			\beta_{2n+1,2k+1}&=2\sqrt{\pi}\frac{\binom{2(k+n+1)}{k+n+1}\binom{2(k-n)}{k-n}}{2^{4k+2}}>0, &\;\;\;0\leq n\leq k.
		\end{aligned}
	\right.
\end{equation*}
\end{lemme}
\begin{proof}[Proof of~\clm{Taylor_expansion_Legendre}]
The proof of this lemma is entirely inspired by the original computation of Legendre\footnote{or Legendre and Laplace, according to a famous paternity controversy.} in his famous \emph{mémoire}~\cite{Leg-1784} where he introduced the polynomials that are nowadays called after him. Let us first rewrite the fraction, for $(0,1)\neq(\theta,t)\in \R\times[0,1]$: $$\frac{1}{\sqrt{1-2t\cos\theta+t^2}}=(1-e^{\i\theta}t)^{-1/2}(1-e^{-\i\theta}t)^{-1/2}.$$ Then, since $\Gamma(1/2-p)=\frac{(-4)^p p!}{(2p)!}\Gamma(1/2)$ and using the following expansion $$(1-x)^{-1/2}=\sum\limits_{p=0}^\infty{\frac{\Gamma(1/2)}{\Gamma(1/2-p)\Gamma(p+1)}(-x)^p}=\sum\limits_{p=0}^\infty{\frac{\binom{2p}{p}}{2^{2p}}x^p},$$ we obtain:
\begin{align*}
\frac{1}{\sqrt{1-2t\cos\theta+t^2}}&=\sum\limits_{(p,q) \in \N^2}{\frac{\binom{2p}{p}\binom{2q}{q}}{2^{2(p+q)}}e^{\i(p-q)\theta}t^{p+q}}\\
	&=\sum\limits_{\substack{k=0 \\ k\textrm{ even}}}^\infty{\sum\limits_{\substack{n=-k \\ n\textrm{ even}}}^k{\frac{\binom{k+n}{(k+n)/2}\binom{k-n}{(k-n)/2}}{2^{2k}}e^{\i n\theta}t^k}}\\
	&\phantom{\leq} +\sum\limits_{\substack{k=1 \\ k\textrm{ odd}}}^\infty{\sum\limits_{\substack{n=-k \\ n\textrm{ odd}}}^k{\frac{\binom{k+n}{(k+n)/2}\binom{k-n}{(k-n)/2}}{2^{2k}}e^{\i n\theta}t^k}}\\
	&=\sum\limits_{k=0}^\infty{\frac{\binom{2k}{k}^2}{2^{4k}}t^{2k}}+\sum\limits_{k=0}^\infty{\sum\limits_{n=1}^k{\frac{\binom{2(k+n)}{k+n}\binom{2(k-n)}{k-n}}{2^{4k}}2\cos(2n\theta)t^{2k}}}\\
	&\phantom{\leq} +\sum\limits_{k=0}^\infty{\sum\limits_{n=0}^k{\frac{\binom{2(k+n+1)}{k+n+1}\binom{2(k-n)}{k-n}}{2^{4k+2}}2\cos((2n+1)\theta)t^{2k+1}}}\\
	&=\sum\limits_{n=0}^\infty{\frac{\binom{2k}{k}^2}{2^{4k}}t^{2k}}+\sum\limits_{n=1}^\infty{\sum\limits_{k=n}^\infty{\frac{\binom{2(k+n)}{k+n}\binom{2(k-n)}{k-n}}{2^{4k}}2\cos(2n\theta)t^{2k}}}\\
	&\phantom{\leq} +\sum\limits_{n=0}^\infty{\sum\limits_{k=n}^\infty{\frac{\binom{2(k+n+1)}{k+n+1}\binom{2(k-n)}{k-n}}{2^{4k+2}}2\cos((2n+1)\theta)t^{2k+1}}}.
\end{align*}
With the definition of the $Y^+_n$'s, this concludes the proof of~\clm{Taylor_expansion_Legendre}.
\end{proof}

Defining all the others $\beta_{p,q}$'s to be zero, this proves~\cpr{Cyl_strict_pos_v_n}: 
\begin{equation*}
v_n(r,r',z-z')=\frac{2}{m^+}\sum\limits_{k=n}^\infty{\beta_{n,k}\left(\frac{m^-}{m^+}\right)^k}>0,
\end{equation*} for $n\geq0$, $r,r'>0$ and $z,z'\in\R$. Moreover, for $\mathbf{r}\neq\mathbf{r'}$, we have $$V(\mathbf{r}-\mathbf{r'})=\sum\limits_{n=0}^\infty{\frac{2}{m^+}\left(\sum\limits_{k=n}^\infty{\beta_{n,k}\left(\frac{m^-}{m^+}\right)^k}\right)Y^+_n(\theta)}.$$
\end{proof}

\begin{rmq}\label{rmq_V_M_first_order_positivity}(The anisotropic potential~\eqref{V_final_definition})
If we define $v_n$ in a similar fashion for the true model based on~\eqref{V_final_definition}, even with the conditions $(\ref{cond_sym_V_M_default}_k)$ and $(\ref{cond_sym_V_M_2_3})$, the $v_n$'s have no sign for $n\geq2$, since we have
$$v_n(r,r',z-z')=\sum\limits_{k=n}^\infty{2\beta_{n,k}\left(\frac{1}{m^+_{Id}}\left(\frac{m^-_{Id}}{m^+_{Id}}\right)^k-\frac{1}{m^+_M}\left(\frac{m^-_M}{m^+_M}\right)^k\right)}$$ which changes sign for $n\geq2$. This is why our method fails if $V$ is given by~\eqref{V_final_definition}. Note that the strict positivity however holds true for $v_0$ and for $v_1$ if $r,r'>0$, which is straightforward using~\eqref{def_coeff_decomp}.
\end{rmq}

As proved in the last Section, $L^2_-(\R^3)$, with corresponding projectors $P^-$, reduces ${\mathfrak L}_Q$. We claim that the spaces $L^2_+(\R_+^*\times\R) \otimes Y^{\sigma}_{n}$, with corresponding projectors $$P^{+}_{n,\sigma}\psi(r,\phi,z)=\left(\displaystyle{\int_0^{2\pi}}{\frac{\psi(r,\phi',z)+\psi(r,\phi',-z)}2Y_n^\sigma(\phi') \dd\phi' }\right)Y_n^\sigma(\phi),$$ also reduce ${\mathfrak L}_Q$. Given that $\left(V\star|Q|^2\right)\in L^2_{\textrm{rad},+}(\R^3)$, $\frac{\dd^2}{\dd\phi^2} Y^{\sigma}_{n}=-n^2Y^{\sigma}_{n}$ and
\begin{equation}\label{laplacien_cyl}
\Delta = \frac{\partial^2}{\partial r^2} + \frac{1}{r}\frac{\partial}{\partial r}+\frac{\partial^2}{\partial z^2} + \frac{1}{r^2}\frac{\partial^2}{\partial \phi^2},
\end{equation} we have
\begin{equation}\label{cyl_decomp_local_term}
\left[-\Delta +1 -\left(V\star|Q|^2\right)\right](fY_n^\sigma)=\left[-\Delta_{(n)}+1-\left(V\star|Q|^2\right)\right](f) Y_n^\sigma,
\end{equation} for any $f\in L^2_+(\R_+^*\times\R)$, and so belonging to $L^2_+(\R_+^*\times\R) \otimes Y^{\sigma}_n$, and where $$-\Delta_{(n)}:=-\frac{\partial^2}{\partial r^2} - \frac{1}{r}\frac{\partial}{\partial r}-\frac{\partial^2}{\partial z^2} + \frac{n^2}{r^2}.$$ Thus the reduction property follows for $-\Delta +1 -\left(V\star|Q|^2\right)$. Moreover, since $V\star(Q\cdot)$ and $P^+_{n,\sigma}$ are linear and using the decomposition $$\psi(r,\phi,z)=c^-(r,\phi,z)+\sum\limits_{n\geq0,\sigma=\pm}{c_{n,\sigma}^+(r,z) Y^{\sigma}_{n}(\phi)},$$ we have to prove that $$V\star(QP^+_{n',\sigma'}c^+_{n,\sigma} Y^{\sigma}_{n}) =P^+_{n',\sigma'}\left(V\star(Qc^+_{n,\sigma} Y^{\sigma}_{n})\right),$$ for any $n,n'\geq0$ and $\sigma,\sigma'=\pm$, in order to conclude. We have
\begin{equation}\label{cyl_decomp_nonlocal_term}
\begin{aligned}
&\left[V\star(Qc^+_{n,\sigma} Y^{\sigma}_{n})\right](r,\phi,z)\\
	&=\displaystyle{\int_{\R_+^*}}{\displaystyle{\int_{-\pi}^{\pi}}{\displaystyle{\int_{\R}}{ Q(r',z')c^+_{n,\sigma}(r',z')Y^{\sigma}_{n}(\phi')V(\mathbf{r}-\mathbf{r'})r' \dd z'} \dd\phi'} \dd r'} \\
	&=\left(\displaystyle{\int_{\R_+^*}}{\displaystyle{\int_{\R}}{ Q(r',z')c^+_{n,\sigma}(r',z')v_n(r,r',z-z')r'  \dd z'} \dd r'}\right)Y^{\sigma}_{n}(\phi).
\end{aligned}
\end{equation} Then using the parity of $v_n$ with respect to its third variable (which is straightforward with~\eqref{def_coeff_decomp}), we obtain $V\star(Qc^+_{n,\sigma} Y^{\sigma}_{n}) \in L^2_+(\R_+^*\times\R) \otimes Y^{\sigma}_{n}$ and the reduction property follows. Thus we can apply~
\cite[Lemma 2.24]{Teschl-09} which gives us that $${\mathfrak L}_Q={\mathfrak L}^-\oplus\bigoplus\limits_{n\geq0, \sigma=\pm}{{\mathfrak L}^+_{n,\sigma} },$$ with ${\mathfrak L}^-={\mathfrak L}_Q^{z-}$ being the same operator as before and each ${\mathfrak L}^+_{n,\sigma}$ a self-adjoint operator on $L^2_+(\R_+^*\times\R) \otimes Y^{\sigma}_n$ with domain $P^+_{n,\sigma}H^2(\R^3)$.
For shortness, we omit the $Q$ subscript in the decomposition ${\mathfrak L}_Q$.

Given~\eqref{cyl_decomp_local_term} and~\eqref{cyl_decomp_nonlocal_term}, for any $n\geq0$ we note ${\mathfrak L}^+_n$ the operator on $L^2_+(\R_+^*\times\R)$ such that ${\mathfrak L}^+_{n,+}(f Y^+_n)={\mathfrak L}^+_n(f) Y^+_n$ and ${\mathfrak L}^+_{n,-}(f Y^-_n)={\mathfrak L}^+_n(f) Y^-_n$. This operator is
\begin{equation*}\label{def_L_n}
{\mathfrak L}^+_n=-\Delta_{(n)}+1+\Phi+W_{(n)}
\end{equation*} where $\Phi$ is the strictly negative multiplication local potential, on $\R_+^*\times\R$, $$\Phi(r,z)=-\left(V\star| Q|^2\right)(r,z)=-\sqrt{2\pi}\displaystyle{\int_{\R_+^*\times\R}}{|Q(r',z')|^2 v_0(r,r',z-z') r'  \dd z' \dd r'}<0$$ and $W_{(n)}$ is the non-local operator, on $\R_+^*\times\R$,
\begin{align}\label{Cyl_W_n_def}
(W_{(n)}f)(r,z)&=-2 Q(r,z)\displaystyle{\int_{\R_+^*\times\R}}{Q(r',z')f(r',z') v_n(r,r',z-z') r'  \dd z' \dd r'}.
\end{align}

Similarly to the non-cylindrical case, we need to prove that $-W_{(n)}$ is positivity improving on $L^2_+(\R_+^*\times\R)$ and this is where the result of~\cpr{Cyl_strict_pos_v_n} is needed.
\begin{lemme}\label{Cyl_W_pos_improv}
For $n\geq0$, the operator $-W_{(n)}$ is positivity improving on $L^2_+(\R_+^*\times\R)$.
\end{lemme}
\begin{proof}[Proof of~\clm{Cyl_W_pos_improv}] Given the definition~\eqref{Cyl_W_n_def} of $-W_{(n)}$, the fact that the $v_n$'s are strictly positive as soon as $r,r'>0$ (by~\cpr{Cyl_strict_pos_v_n}) and that $Q>0$, it follows that $-W_{(n)}$ is positivity improving on $L^2_+(\R_+^*\times\R)$ for any $n\geq0$.
\end{proof}

\subsubsection{Perron--Frobenius property}\label{Perron-Frobenius_property_cyl}
We now prove that the ${\mathfrak L}^+_n$'s verify the Perron--Frobenius property.

\begin{prop}[Perron--Frobenius properties]\label{cyl_Perron_fro_ppty}For $n>0$, the ${\mathfrak L}^+_n$'s are essentially self-adjointness on $C_0^\infty(\R_+^*\times\R)$ and bounded below.

Moreover they have the Perron--Frobenius property: if $\lambda^{n}_0$ denotes the lowest eigenvalue of ${\mathfrak L}^+_n$, then $\lambda^{n}_0$ is simple and the corresponding eigenfunction $\psi^{n}_0$ is strictly positive.
\end{prop}
\begin{proof}[Proof of~\cpr{cyl_Perron_fro_ppty}] We follow the structure of the proof of \cite[Lemma 8]{Lenzmann-09}.

\medskip

\noindent\textbf{Self-adjointness.}  We still have $V\star \left| Q\right|^2 \in L^4(\R^3)\cap L^\infty(\R^3)$. Moreover, defining $\mathring{f}(r,\cdot,z)=f(r,z)Y_n^+\in L^2_+(\R_+^*\times\R)\otimes Y_n^+\subset L^2(\R^3)$, for any $f\in L^2_+(\R_+^*\times\R)$, we have $\pscal{f,g}_{L^2_+(\R_+^*\times\R)}={\langle \mathring{f},\mathring{g} \rangle}_{L^2(\R^3)}$ and, consequently, that $\Phi+1$ is a bounded operator on $L^2_+(\R_+^*\times\R)$. Then, for any $\xi\in L^2_+(\R_+^*\times\R)$ and $p\in[2,\infty]$, we have $$\norm{W_{(n)}\xi}_{L^p(\R_+\times\R)}\leq \norm{Q}_{L^p}\normSM{V\star( Q\mathring{\xi})}_{L^\infty}.$$ Thus $W_{(n)}\xi\in L^2_+(\R_+^*\times\R)\cap L^\infty(\R_+^*\times\R)$ and, finally, $1+\Phi+W_{(n)}$ and, thus, ${\mathfrak L}^+_n$ are bounded below on $L^2_+(\R_+^*\times\R)$. Furthermore, it is known that $-\Delta_{(n)}$ is essentially self-adjoint on $C_0^\infty(\R_+^*\times\R)$ provided that $n>0$. Thus, given that $1+\Phi+W_{(n)}$ is bounded (so $-\Delta_{(n)}$-bounded of relative bound zero), symmetric (moreover self-ajoint) and that its domain contains the domain of $-\Delta_{(n)}$, we obtain by the Rellich-Kato theorem the essentially self-adjointness of ${\mathfrak L}^+_n$ on $C_0^\infty(\R_+^*\times\R)$.

\medskip

\noindent\textbf{Positivity improving.} We claim that $e^{t\Delta_{(n)}}$ is positivity improving for all $t>0$ on $L^2(\R_+^*\times\R)$. Indeed we have the explicit formula for the integral kernel $e^{t\Delta}$ on $\R^3$, namely,
\begin{equation}\label{heat_Kernel_laplacian_cyl}
e^{t\Delta}(x,y)=(4\pi t)^{-3/2} e^{-\frac{|x-y|^2}{4t}}=(4\pi t)^{-3/2} e^{-\frac{r^2+r'^2+(z-z')^2}{4t}} e^{\frac{rr'}{2t}\cos(\phi-\phi')},
\end{equation} for all $x:=(r,\phi,z)$ and $y:=(r',\phi',z')$. On the other hand we have
\begin{equation}\label{exp_and_modified_bessel_first_kind}
e^{x \cos\theta} = \sqrt{2\pi}\sum_{m=0}^\infty  \sqrt{2}^{\delta_{\{m\geq1\}}} I_m(x) Y_m^+(\theta),\quad \forall x\in\R,
\end{equation} where $I_n(x)=\pi^{-1}\smallint_0^\pi{\exp(x\cos\theta)\cos(n\theta) \dd\theta}$ are the modified Bessel functions of the first kind, that are strictly positive for $n\geq0$ and $x>0$. From these two relations, we deduce the integral kernel $e^{t\Delta_{(n)}}$ acting on $L^2(\R_+\times\R)$ and that it is positive, which are the two points of the following lemma.
\begin{lemme}\label{Computation_integral_kernel_cylindrical_laplacian} Let $f\in L^2(\R_+^*\times\R)$, $r,r'>0$ and $n\geq0$. Then the integral kernel $e^{t\Delta_{(n)}}$ acting on $L^2(\R_+^*\times\R)$ verifies
\begin{equation}\label{heat_Kernel_cylindrical_laplacian}
		e^{t\Delta_{(n)}}((r,z),(r',z'))=\frac{\sqrt{2}^{-\delta_n^0}}{4\pi t^{3/2}} e^{-\frac{r^2+r'^2+z^2+z'^2}{4t}}I_n\left(\frac{rr'}{2t}\right)\exp\left(\frac{zz'}{2t}\right)>0.
\end{equation}
\end{lemme}
\begin{proof}[Proof of~\clm{Computation_integral_kernel_cylindrical_laplacian}]
Let $f\in L^2(\R_+^*\times\R)$. Using \eqref{heat_Kernel_laplacian_cyl}, for $n\geq0$, we have
\begin{align*}
	&\left(e^{t\Delta}(fY_n^\sigma)\right)(r,\phi,z)\\
	&=(4\pi t)^{-3/2}\displaystyle{\int_{\R_+\times\R}}{e^{-\frac{r^2+r'^2+(z-z')^2}{4t}}f(r',z') \left(\displaystyle{\int_{-\pi}^\pi}{e^{\frac{rr'}{2t}\cos(\phi-\phi')}Y_n^\sigma(\phi') \dd\phi'}\right)r'\dd r' \dd z' }\\
	&=\frac{\sqrt{2}^{-\delta_n^0}}{4\pi t^{3/2}}\displaystyle{\int_{\R_+\times\R}}{e^{-\frac{r^2+r'^2+z^2+z'^2}{4t}}f(r',z')I_n\left(\frac{rr'}{2t}\right)\exp\left(\frac{zz'}{2t}\right)r'\dd r' \dd z' } Y_n^\sigma(\phi).
\end{align*} Which allows to conclude the proof of~\clm{Computation_integral_kernel_cylindrical_laplacian}.
\end{proof} So, for all $n\geq0$, $e^{t\Delta_{(n)}}$ is positivity improving on $L^2(\R_+^*\times\R)$ for all $t>0$ and consequently on $L^2_+(\R_+^*\times\R)$. Then, by functional calculus, we have $$\left(-\Delta_{(n)} + \mu\right)^{-1}=\displaystyle{\int_0^\infty{e^{-t\mu}e^{t\Delta_{(n)}} \dd t}}, \;\;\; \forall\,\mu>0,$$ thus $(-\Delta_{(n)} + \mu)^{-1}$ is positivity improving on $L^2_+(\R_+^*\times\R)$ for all $\mu>0$.

Moreover we claim that $-(\Phi + W_{(n)})$ is positivity improving on $L^2_+(\R_+^*\times\R)$ since $-\Phi$ is a positive multiplication operator and $-W_{(n)}$ is positivity improving on $L^2_+(\R_+^*\times\R)$.

\medskip

The end of the proof uses the exact same arguments as in the proof of the Perron--Frobenius property in the non-cylindrical case (\cpr{symm_Perron_fro_ppty}) and, consequently, this ends the proof of~\cpr{cyl_Perron_fro_ppty}.
\end{proof}

\subsubsection{Proof of \cth{THM_Ker_L}}
First, using the results of the previous Section, we have $$\ker \left({\mathfrak L}_Q\right)_{|L^2_-(\R^3)}=\ker \left({\mathfrak L}_Q\right)_{|L^2_{-,+,+}(\R^3)}=\vect\{\partial_z Q\}$$ and ${\mathfrak L}_Q\partial_x Q\equiv {\mathfrak L}_Q\partial_y Q\equiv0$. But now $Q$ is furthermore \emph{cylindrical-even}, thus $\partial_x Q=\frac{x}{r}\partial_r Q \in L^2_+(\R_+^*\times\R) \otimes Y^+_1$ and $\partial_y Q=\frac{y}{r}\partial_r Q \in L^2_+(\R_+^*\times\R) \otimes Y^-_1$, which implies that ${\mathfrak L}^+_1\partial_r Q\equiv0$. Then, $Q>0$ being \emph{cylindrical-even strictly decreasing}, $\partial_r Q<0$ on $\R_+\times\R$ and, by the Perron--Frobenius property (\cpr{cyl_Perron_fro_ppty}), it is (up to sign) the unique eigenvector associated with the lowest eigenvalue of ${\mathfrak L}^+_1$, namely $\lambda^{1}_0=0$. Consequently, $\partial_x Q$ (resp. $\partial_y Q$) is the unique eigenvector --- up, in addition, to rotations in the cylindrical plane --- associated with the lowest eigenvalue $\lambda^{1,+}_0=0$ (resp. $\lambda^{1,-}_0=0$) of ${\mathfrak L}^+_{1,+}$ (resp. ${\mathfrak L}^+_{1,-}$). To summarize, we know at this point that $$\ker {\mathfrak L}_Q=\vect\left\{\partial_x Q, \partial_y Q, \partial_z Q\right\} \oplus \ker \left({\mathfrak L}_Q\right)_{|L^2_{\textrm{rad},+}(\R^3)}\oplus\bigoplus\limits_{\substack{n\geq2\\\sigma=\pm}}{ \ker \left({\mathfrak L}_Q\right)_{|L^2_+(\R_+^*\times\R) \otimes Y^{\sigma}_{n}}},$$ and we have to deal with the higher orders. The end of the paper is devoted to the proof that
\begin{equation}\label{triviality_kernel_others}
\ker {\mathfrak L}^+_{n,\sigma} =\{0\},\;\; \forall n\geq2, \sigma=\pm.
\end{equation} For $n\geq2$, let $0<\phi^n\in L^2_+(\R_+^*\times\R)$ be an eigenvector of ${\mathfrak L}^+_n$ associated with $\lambda^n_0$. Then $\phi^nY^+_n$ (resp. $\phi^nY^-_n$) is an eigenvector of ${\mathfrak L}^+_{n,+}$ (resp. ${\mathfrak L}^+_{n,-}$) associated to the eigenvalue $\lambda^{n,+}_0=\lambda^n_0$ (resp. $\lambda^{n,-}_0=\lambda^n_0$). Thus, for $n\geq2$ and $\sigma=\pm$, we have
\begin{align*}
\lambda^{1,\sigma}_0-\lambda^{n,\sigma}_0&\leq\pscal{\phi^{n},{\mathfrak L}^+_{1} \phi^{n}}_{L^2(\R_+\times\R)}-\pscal{\phi^{n},{\mathfrak L}^+_{n} \phi^{n}}_{L^2(\R_+\times\R)}\\
	&\leq-\displaystyle{\int_{\R_+\times\R}}{\frac{n^2-1}{r^2}\left(\phi^{n}(r,z)\right)^2r\dd r\dd z} \\
	&+2\!\!\displaystyle{\iint\limits_{\left(\R_+\times\R\right)^2}}{\![Q\phi^{n}](r,z)[Q\phi^{n}](r',z')\left[v_n-v_{1}\right](r,r',z-z') rr'  \dd z\dd z' \dd r \dd r'}.
\end{align*} Since $Q>0$ and $\phi^n>0$ (by the Perron--Frobenius property in~\cpr{cyl_Perron_fro_ppty}), in order to prove that $\lambda^{n,\sigma}_0>\lambda^{1,\sigma}_0$, it is sufficient to prove that $v_n<v_1$ almost everywhere on $(\R_+^*\times\R)^2$. Using the explicit formula~\eqref{V_explicit} for $V$, this is equivalent to prove that $$T_n:=\displaystyle{\int_{0}^\pi}{\frac{\cos(n\theta)-\cos\theta}{\sqrt{K+2a_2^2rr'(1-\cos\theta)}}\dd \theta}<0, \quad \textrm{for a.e. } (r,z),(r',z')\in\left(\R^*_+\times\R\right)^2,$$ where $K=|(1-S)^{-1}(r-r',0,z-z')|^2$. First, let us remark that the points $\left\{2\frac{k}{n-1}\pi\right\}_{k\in\Z}$ and $\left\{2\frac{k}{n+1}\pi\right\}_{k\in\Z}$ are the zeros of $\cos(n\cdot)\textcolor{blue}{-}\cos(\cdot)$ and that the function $g=\left[K+2(1-s_2)^{-2}rr'(1-\cos(\cdot))\right]^{-1/2}$ is strictly decreasing on $]0,\pi[$. Let us define, $\theta_{2\lfloor n/2\rfloor}:=\pi$ and, for $k$ an integer in $[0,\lfloor n/2\rfloor-1]$, $\theta_{2k}:=2\frac{k}{n-1}\pi$ and $\theta_{2k+1}:=2\frac{k+1}{n+1}\pi$ which are all the zeros of $\cos(n\cdot)-\cos(\cdot)$ in $[0,\pi]$, except $\theta_{2\lfloor n/2\rfloor}$ if $n$ is even. Then, noticing that $\cos(n\cdot)-\cos(\cdot)$ is strictly negative on intervals $]\theta_{2k},\theta_{2k+1}[$, strictly positive on intervals $]\theta_{2k+1},\theta_{2k+2}[$ and that $n\theta_{2k}=2k\pi+\theta_{2k}$, we have
\begin{align*}
T_n&=\sum\limits_{k=0}^{\lfloor\frac{n}{2}\rfloor-1}{\displaystyle{\int_{\theta_{2k}}^{\theta_{2k+1}}}{\!\!\!\!\!\!\!\!\!\underset{>g(\theta_{2k+1})>0}{\underbrace{g(\theta)}}\underset{<0}{\underbrace{\left(\cos(n\theta)-\cos\theta\right)}}\dd \theta}+\displaystyle{\int_{\theta_{2k+1}}^{\theta_{2k+2}}}{\!\!\!\!\underset{0<\cdot< g(\theta_{2k+1})}{\underbrace{g(\theta)}}\underset{>0}{\underbrace{\left(\cos(n\theta)-\cos\theta\right)}}\dd \theta}} \\
	&<\sum\limits_{k=0}^{\lfloor\frac{n}{2}\rfloor-1}{g(\theta_{2k+1})\displaystyle{\int_{\theta_{2k}}^{\theta_{2k+2}}}{\left(\cos(n\theta)-\cos\theta\right)\dd \theta}}.
\end{align*} If $n=2$ or $n=3$, this immediately leads to $T_n<0$. Otherwise, if $n\geq4$, we have
\begin{align*}
T_n&<\sum\limits_{k=0}^{\lfloor\frac{n}{2}\rfloor-1}{g(\theta_{2k+1})\displaystyle{\int_{\theta_{2k}}^{\theta_{2k+2}}}{\left(\cos(n\theta)-\cos\theta\right)\dd \theta}} \\
	&<-\frac{n-1}{n}\bigg(\sum\limits_{k=0}^{\lfloor\frac{n}{2}\rfloor-2}{g(\theta_{2k+1})\sin\theta_{2k+2}}-\sum\limits_{k=1}^{\lfloor\frac{n}{2}\rfloor-1}{g(\theta_{2k+1})\sin\theta_{2k}}\bigg) \\
	&<-\frac{n-1}{n}\sum\limits_{k=1}^{\lfloor\frac{n}{2}\rfloor-1}{\underset{>0}{\underbrace{\left[g(\theta_{2k-1}-g(\theta_{2k+1})\right]}}\underset{>0}{\underbrace{\sin\theta_{2k}}}}<0.
\end{align*}

 Thus we have just proved, for $n\geq2$ and $\sigma=\pm$, that $\lambda^{n,\sigma}_0>\lambda^{1,\sigma}_0=0$, consequently $\ker {\mathfrak L}^+_{n,\sigma} =\{0\}$.

This concludes the proof of~\cth{THM_Ker_L}.
\hfill\qedsymbol


\end{document}